\documentclass{amsart}

\usepackage{overpic}

\newtheorem{result}{Main Result}
\newtheorem{theorem}{Theorem}
\newtheorem{definition}{Definition}
\newtheorem{thm}{Theorem}[section]

\newtheorem{prop}[thm]{Proposition}

\newtheorem{corollary}{Corollary}[theorem]
\theoremstyle{definition}

\newtheorem{example}[thm]{Example}

\theoremstyle{remark}

\begin{document}

\title[Mechanisms for Network Growth]{Mechanisms for Network Growth that Preserve Spectral and Local Structure}


\author{L. A. Bunimovich$^1$}

\author{B. Z. Webb$^2$}

\maketitle

\begin{center}
$^{1}$ School of Mathematics, Georgia Institute of Technology, 686 Cherry Street, Atlanta, GA 30332\\
$^{2}$ Department of Mathematics, Brigham Young University, TMCB 308, Provo, UT 84602\\
E-mail: $^1$bunimovich@math.gatech.edu and $^2$bwebb@mathematics.byu.edu
\end{center}

\begin{abstract}
We introduce a method that can be used to evolve the topology of a network in a way that preserves both the network's spectral as well as local structure. This method is quite versatile in the sense that it can be used to evolve a network's topology over any collection of the network's elements. This evolution preserves both the eigenvector centrality of these elements as well as the eigenvalues of the original network. Although this method is introduced as a tool to model network growth, we show it can also be used to compare the topology of different networks where two networks are considered similar if their evolved topologies are the same. Because this method preserves the spectral structure of a network, which is related to the network's dynamics, it can also be used to study the interplay of network growth and function. We show that if a network's dynamics is \emph{intrinsically stable}, which is a stronger version of the standard notion of stability, then the network remains intrinsically stable as the network's topology evolves. This is of interest since the growth of a network can have a destabilizing effect on the network's dynamics, in general. In this sense the methods developed here can be used as a tool for designing mechanisms of network growth that ensure a network remains stabile as it grows.
\end{abstract}

\section{Introduction}

Networks in the biological, social, and technological sciences fill specific needs. Gene regulatory networks play a central role in cell growth, neural networks are responsible for certain cognitive processes, and metabolic networks determine the biochemical properties of cells \cite{KS08,BS09,BO04}. Social networks such as the interactions of social insects \cite{CF14,HEOGF13} and professional sports teams \cite{FAIPW12} are used to achieve a common goal or outcome. Technological networks such as the internet together with the World Wide Web allow access to information.

The task these networks are able to perform depend on two common features. First, each network is built out of a number of components that typically have a complicated structure of interactions, which we refer to as the network's \emph{topology}. Second, the vast majority of these networks are \emph{dynamic} in the sense that each network component has a certain behavior that depends on those components it interacts with. The pattern of behavior that emerges from these interactions is the network's \emph{dynamics}.

The extent to which a network can perform a certain task typically depends both on the network's topology and dynamics. For example, an electrical grid with a well-designed structure of power lines, substations, etc. would likely be considered substandard if it only provided power intermittently. Alternatively, if power was supplied at all times but the design of the power grid limited the number of customers, the same would be true. In this way both the topology and dynamics of a network are important to the network's ability to perform its function.

To complicate things, real-world networks are not only dynamic in terms of the behavior of their elements but also in terms of their topology. For example, neural networks are continually adding new neurons and connections to process and store incoming information. Social networks are constantly reorganizing themselves as new relationships are formed and old ones are dissolved. Similarly, the World Wide Web has an ever changing structure of interactions as web pages are updated, added, and deleted (see \cite{GS09} for a review of the changing topology of networks).

For the sake of clarity, throughout the paper we will refer to a network's \emph{evolution} as the changes that occur in the network's topology over time. We will use the term \emph{dynamics} when referring to the collective behavior of the elements that make up the network, although both the changes in the network's topology and the behavior of these elements are in fact aspects of the network's dynamics.

One of the central questions in network science is how a complex system, i.e. a network, with an evolving structure of interactions can maintain a specific function. For instance, a beating heart maintains its dynamics even as the cellular structure grows over time. Similarly, electrical grids need to be engineered in a way that allows for power to be supplied even as new power lines, stations, etc. are added to the existing grid.

The issue is that network growth, although important to the network's function, can have a destabilizing effect on the network's dynamics, which in turn can lead to network failure. Cancer, for instance, which is the uncontrolled growth of unwanted cells is an example of this phenomena in biological networks.

In this paper we develop a theory that offers a flexible method for modeling the growth of a network. This method is built around the theory of isospectral network transformations \cite{BW12,BWBook}, which describes how certain changes in a network's structure effect the network's \emph{spectrum}, i.e. the eigenvalues associated with the network. Since the eigenvalues of a network are related to both the network's structure and dynamics, this allows us to study how changes in the graph's structure effect the network's spectrum and how this impacts the network's dynamics and ultimately its function.

To be resilient to failures networks typically have a certain amount of redundancy in their topology \cite{MSA08,TSE99}. Hence, as the topology of a network evolves the network is likely evolves in ways that preserves certain structures useful to the network. Such statistically significant structures are often referred to as network motifs \cite{A07}. Using this idea we evolve the topology of a network in a way that preserves to a large extent the network's \emph{local structure}, which has the effect of largely preserving the spectral properties of the network.

To evolve the topology of a network we begin by describing its structure of interactions by a graph $G$ whose vertices $V$ (nodes) represent the network's elements (components) and whose edges $E$ represent the network's interactions. The graph $G$ is evolved by selecting a subset of the network's elements $S\subset V$, which we will refer to as the \emph{core} of the evolved graph.

To evolve the graph $G$ we introduce the notion of a \emph{branch of components}. The components we consider are the strongly connected components of the graph made up of those vertices not in $S$. A branch of components consists of either a path or cycle of these components that begins and ends at a vertex of $S$ (see definition \ref{def:componentbranch}). By merging these branches together the result is the evolved graph $\mathcal{X}_S(G)$ (see definition \ref{def:exp}), which represents the network's topology at some later point in time.

Because this type of transformation preserves the components of the original graph, the evolved graph is in this sense \emph{locally indistinguishable} from the original graph. The difference is that the evolved graph will typically have many more of these components than the original graph. This difference in structure is reflected in the difference in the spectrum $\sigma(G)$ of the original graph $G$ and spectrum $\sigma(\mathcal{X}_S(G))$ of the evolved graph $\mathcal{X}_S(G)$. This is one of the paper's main results and is described by the following theorem (see section \ref{sec2}, theorem \ref{thm1}).

\begin{result}\textbf{(Spectra of Evolved Graphs)}
Let $G$ be a graph with vertex set $V$. If $S\subseteq V$ let $C_1,\dots,C_m$ be the strongly connected components of $G|\bar{S}$ where $\bar{S}$ is the complement of $S$. Then
\[
\sigma\big(\mathcal{X}_S(G)\big)=\sigma(G)\cup\sigma(C_1)^{n_1-1}\cup\sigma(C_2)^{n_2-1}\cup\dots\cup \sigma(C_m)^{n_m-1}
\]
where $n_i$ is the number of components $C_i$ in the evolved graph $\mathcal{X}_S(G)$ and $\sigma(C_i)^{n_i-1}$ denotes $n_i-1$ copies of the eigenvalues of $C_i$.
\end{result}

The spectrum of the evolved graph $\mathcal{X}_S(G)$ is then the spectrum of the original graph $G$ together with some collection of eigenvalues of the components $C_1,\dots,C_m$ of $G$. Therefore, if a network has a changing topology that can be modeled via a graph evolution then we can effectively predict changes in both the network's topology and spectrum over time. Since the dynamics of a network depends on its spectrum, the theory of graph evolutions developed in this paper is a potential tool to study the interplay of network growth and function.

A graph evolution also preserves the eigenvectors of graph in a certain way. Specifically, a graph evolution preserves the part of the graph's eigenvectors that correspond to the network's core under some mild conditions (see section \ref{sec2}, proposition \ref{prop:0}). An important consequence of this fact is that the eigenvector centrality of the network's core is preserved as the network evolves (see section \ref{sec2}, theorem \ref{prop10}).

\begin{result}\textbf{(Eigenvector Centrality of Evolved Graphs)}
Let $G=(V,E)$ be strongly connected with its eigenvector centrality given by the vector $\mathbf{p}$. If $S\subset V$ then the eigenvector centrality of $\mathcal{X}_S(G)$ is given by a vector $\mathbf{q}$ where $\mathbf{p}_S=\mathbf{q}_S$. The vectors $\mathbf{p}_S$ and $\mathbf{q}_S$ are respectively the vectors $\mathbf{p}$ and $\mathbf{q}$ restricted to the entries indexed by $S$.
\end{result}

Another important feature of a graph evolution is that it is extremely versatile in the sense that a graph $G$ can be evolved with respect to any subset of its vertex set. One of the applications of this fact, which we explore in this paper, is to use graph evolutions to determine whether two related or unrelated networks with graphs $G$ and $H$ are similar in the following way.

Is it is fairly straight-forward to devise a rule $\tau$ that selects a unique vertex set of any graph. For instance, the rule $\tau$ that selects those vertices with highest degree, eigenvalue centrality, clustering coefficient, etc. is such a rule, which we refer to as a \emph{structural rule}. Two graphs $G$ and $H$ are considered to be \emph{similar} to each other if they evolve into the same graph under the rule $\tau$, which we write as $\tau(G)\simeq\tau(H)$, where $\tau(G)$ is the graph $G$ evolved with respect to $\tau$. It turns out that this notion of similarity, which we refer to as \emph{evolution equivalence}, can be used to partition the space of all graphs into those graphs that are similar, i.e. are evolution equivalent, with respect to $\tau$ and those that are not. This is summarized in the following result (see section \ref{sec3}, theorem \ref{thm2}).

\begin{result}\textbf{(Evolution Equivalence)}
Suppose $\tau$ is a rule that selects a unique set of vertices from any graph, i.e. a structural rule. Then $\tau$ induces an equivalence relation $\sim$ on the set of all graphs where $G\sim H$ if $\tau(G)\simeq\tau(H)$.
\end{result}

One reason for designing such a rule $\tau$ is that most of the time it is not obvious that two different graphs are equivalent. That is, two graphs may be similar but until both graphs are evolved with respect to some rule $\tau$ this may be difficult to see. In this paper we show that by choosing an appropriate rule $\tau$ one can discover this similarity (see examples \ref{ex:evoequ} and \ref{ex:semiequ}).

In this way, a rule $\tau$ allows those studying a particular class of networks a way of comparing the \emph{evolved topology} of these networks and drawing conclusions about both the evolved and original networks. Of course, it is important that the rule $\tau$ be designed by the particular biologist, chemist, physicist, etc. to have some significance with respect to the nature of the networks under consideration.

In fact, the rule used can be completely arbitrary, i.e. the biologist, chemist, physicist can choose whatever set of elements she or he deems important and disregard the notion of choosing a rule altogether. Although having no fixed rule means that we lose the notion of evolution equivalence, it is still possible to compare the structure of the evolved graphs. What is important is that many networks currently under study are likely to have features that come to light as these networks are evolved.

Not only can a rule $\tau$ be used to discover the similarities between two graphs and their associated networks but this rule can be used to \emph{sequentially evolve} the structure of a graph $G$, which results in the sequence
\[
G, \ \tau(G), \ \tau^2(G), \ \tau^3(G),\dots, \ \tau^i(G),\dots
\]
This allows one to study the long-term or \emph{asymptotic evolution} of a network's topology under $\tau$. As opposed to many of the most well-known methods for evolving the topology of a graph, such as preferential attachment \cite{BA02}, the sequence of evolutions is deterministic and results in the unique graph $\tau^i(G)$ after each step. Moreover, because of the way in which a graph evolution is defined, the graph $\tau^i(G)$ becomes sparser after each iteration, which is important as most real networks are sparse \cite{N03,HG08}.

To demonstrate how this theory of graph evolutions can be used in the study of network dynamics we consider a class of \emph{dynamical networks}, which are dynamical systems with an underlying graph structure, that can be evolved with respect to a given rule $\tau$. The specific dynamical property we consider here is stability, which is observed in a number of important systems including neural networks \cite{Cao2003,Cheng2006,SChena2009,MCohen1983,LTao2011}, epidemic models \cite{Wang2008}, and in the study of congestion in computer networks \cite{Alpcan2005}.

We show that if a dynamical network $(F,X)$ is intrinsically stable, which is a stronger form of stability than the standard notion of stability (see definition \ref{def:intrinsic}), then the evolved version of this network $(F_{\tau},X_{\tau})$ under any rule $\tau$ remains intrinsically stable. This is summarized in the following theorem (see section \ref{sec:4}, theorem \ref{thm:evostability}).

\begin{result}\textbf{(Stability of Structurally Evolving Networks)}
Let $(F,X)$ be a dynamical network and $\tau$ a structural rule. If $(F,X)$ is intrinsically stable then the evolved network $(F_{\tau},X_{\tau})$ is also intrinsically stable.
\end{result}

Additionally, we show that stability by itself is not enough to guarantee that a network will remain stable as its topology evolves (see example \ref{ex:loss}). That is, a network may lose its stability as it grows if it is not intrinsically stable. That is, intrinsic stability provides a way of designing stable networks that maintain stability even as their topology evolves.

The paper is organized as follows. In section \ref{sec2} we introduce the notion of a graph evolution and prove our first main result regarding the spectrum of an evolved graph (theorem \ref{thm1}). In this section we also describe how a graph evolution effects the graph's eigenvectors and in particular preserves the eigenvector centralities of an associated network. In section \ref{sec3} we use this theory together with the idea of a structural rule $\tau$ to develop the concept of a evolution equivalence. That is, we show that any rule $\tau$ can be used to compare and analyze the structure of networks in a variety of ways (theorem \ref{thm2}).

In section \ref{sec:4} we use graph evolutions to evolve the structure of a dynamical system used to model network dynamics. Here we show that if such a network is intrinsically stable then an evolved version of the network is also stable (theorem \ref{thm:evostability}). Section \ref{conc} contains some concluding remarks. Section \ref{appendix} contains a proof of theorem \ref{thm1} together with the necessary parts of the theory of isospectral network transformations needed to prove this result.

\section{Structural Evolution of Graphs}\label{sec2}

The standard method used to describe the topology of a network is a graph. A \emph{graph} $G=(V,E,\omega)$ is composed of a \emph{vertex set} $V$, an \emph{edge set} $E$, and a function $\omega$ used to weight the edges $E$ of the graph. The vertex set $V$ represents the \emph{elements} of the network, while the edges $E$ represent the links or \emph{interactions} between these network elements. The weights of the edges given by $\omega$ measure the \emph{strength} of these interactions.

In general, the edges $E$ of a graph can either be \emph{directed} or \emph{undirected}. For instance, the edges in a graph representing the World Wide Web are directed since hyperlinks take a user from one page to another. The internet, which is a physical network of data connections between computers, is represented by a graph that has undirected edges since such connections can transfer data in both directions.

The techniques and results we present are valid for both directed as well as undirected graphs. To describe these both these techniques and results it is worth emphasizing that directed graphs are more general than undirected graphs. The reason is that an undirected graph can be considered to be a directed graph by replacing each of its edges by two directed edges that point in opposite directions. Similarly, weighted graphs are more general than unweighted graphs since an unweighted graph can be made into a weighted graph by giving each of its edges unit weight.

With this in mind, the graphs we consider in this paper are formally those graphs that are either directed or undirected and either weighted or unweighted. However, as any such graph can be considered to be a weighted directed graph then, without loss in generality, we consider those graphs $G=(V,E,\omega)$ that are both weighted and directed. That is, we let $V=\{v_1,\dots,v_n\}$, where $v_i$ represents the $i$th network element. We let $e_{ij}$ denote the edge that begins at $v_i$ and ends at $v_j$. In terms of the network, the edge $e_{ij}$ belongs to the edge set $E$ if the $i$th network element directly influences or is linked to the $j$th network element.

One of the most natural ways of investigating a graph is to analyze its path and cycle structure. This approach goes back to the origins of graph theory where Euler used these ideas to solve the K{\"o}nigsberg bridge problem \cite{Alex06}. A \emph{path}\index{path} $P$ in the graph $G=(V,E,\omega)$ is an ordered sequence of distinct vertices $P=v_1,\dots,v_m$ in $V$ such that $e_{i,i+1}\in E$ for $i=1,\dots,m-1$. If the vertices $v_1$ and $v_m$ are the same then $P$ is a \emph{cycle}\index{cycle}. If it is the case that a cycle contains a single vertex then we call this cycle a \emph{loop}\index{loop}.

A fundamental idea related to the structure of a graph is the notion of a strongly connected component. The idea is that there may be a path from $v_i$ to $v_j$ but no path from $v_j$ to $v_i$. In this case the $j$th vertex can be reached from the $i$th but not the other way around. In a strongly connected component every vertex can be reached from every other vertex and so every element in this component of the associated network can have an effect on every other element in the component.

A graph $G=(V,E,\omega)$ is \emph{strongly connected} if for any vertices $v_i,v_j\in V$ there is a path from $v_i$ to $v_j$ or $G$ consists of a single vertex. A \emph{strongly connected component} of a graph $G$ is a subgraph that is strongly connected and is maximal with respect to this property.

Because we are concerned with evolving the structure of a network in ways that preserve, at least locally, the structure of a graph, we will need the notion of a graph restriction. For a graph $G=(V,E,\omega)$ and a subset $S\subseteq V$ we let $G|S$ be the \emph{restriction} of the graph $G$ to the vertex set $S$, which is the subgraph of $G$ on the vertex set $S$ along with any edges of $E$ between vertices in $S$. Importantly, we let $\bar{S}$ denote the complement of $S$, so that the restriction $G|\bar{S}$ is the graph restricted to the complement of those vertices not in $S$.

The key to evolving the structure of a graph in a way that preserves to a large extent both the spectral properties of the graph and the graph's \emph{local structure} is to look at the strongly connected components of the restricted graph $G|\bar{S}$. If $C_1,\dots,C_m$ denote these strongly connected components then our goal is to find paths or cycles of these components, which we refer to as branches of components.

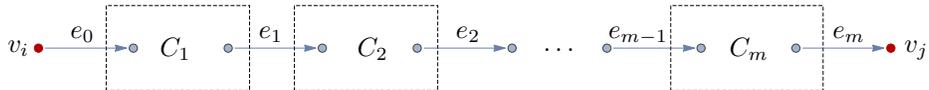
\begin{figure}
\begin{center}
\begin{tabular}{c}
    \begin{overpic}[scale=.275]{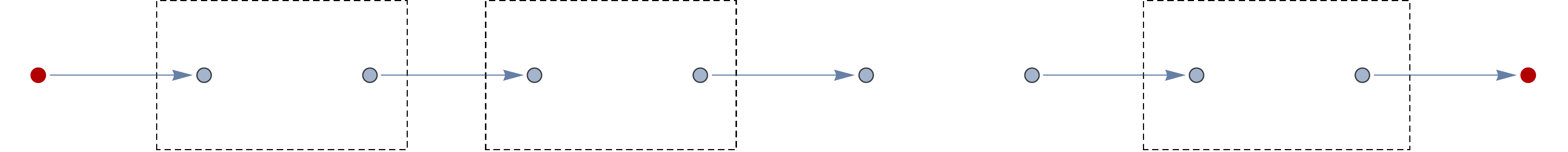}
    \put(-1,4.5){$v_i$}
    \put(6,6){$e_0$}
    \put(16,4){$C_1$}
    \put(27,6){$e_1$}
    \put(38,4){$C_2$}
    \put(49,6){$e_2$}
    \put(58.75,4.55){$\Huge \dots$}
    \put(66,6){$e_{m-1}$}
    \put(79.5,4){$C_m$}
    \put(91,6){$e_m$}
    \put(99,4.5){$v_j$}
    \end{overpic}
\end{tabular}
\end{center}
  \caption{A representation of a path of components is shown, consisting of the sequence $C_1,\dots,C_m$ of components beginning at vertex $v_i$ and ending at vertex $v_j$. From $C_k$ to $C_{k+1}$ there is a single directed edge $e_{k+1}$. From $v_i$ to $C_1$ and from $C_m$ to $v_j$ there is also a single directed edge.}\label{fig01}
\end{figure}

\begin{definition}\label{def:componentbranch} \textbf{(Component Branches)}
For a graph $G=(V,E,\omega)$ and vertex set $S\subseteq V$ let $C_1,\dots,C_m$ be strongly connected components of $G|\bar{S}$. If there are edges $e_0,e_1,\dots,e_m\in E$ and two vertices $v_i,v_j\in S$ such that\\
(i) $e_k$ is an edge from a vertex in $C_k$ to a vertex in $C_{k+1}$ for $k=1,\dots,m-1$;\\
(ii) $e_0$ is an edge from $v_i$ to a vertex in $C_1$; and\\
(iii) $e_m$ is an edge from a vertex in $C_m$ to $v_j$, then we call the sequence
\[
\beta=v_i,e_{0},C_1,e_{1},C_2,\dots,C_m,e_{m},v_{j}
\]
a \emph{path of components} of $G$ with respect to $S$. In the case that $v_i=v_j$ then $\beta$ is a \emph{cycle of components}. We call the collection $\mathcal{B}_S(G)$ of these paths and cycles the \emph{component branches} of $G$ with respect to $S$.
\end{definition}

A representation of the path of components described in definition \ref{def:componentbranch} is shown in figure \ref{fig01}. The sequence of components $C_1,\dots,C_m$ in this definition  can be empty in which case $m=0$ and $\beta$ is the path $\beta=v_1,v_2$ or loop if $v_1=v_2$.











\begin{figure}
\begin{center}
\begin{tabular}{cc}
    \begin{overpic}[scale=.18]{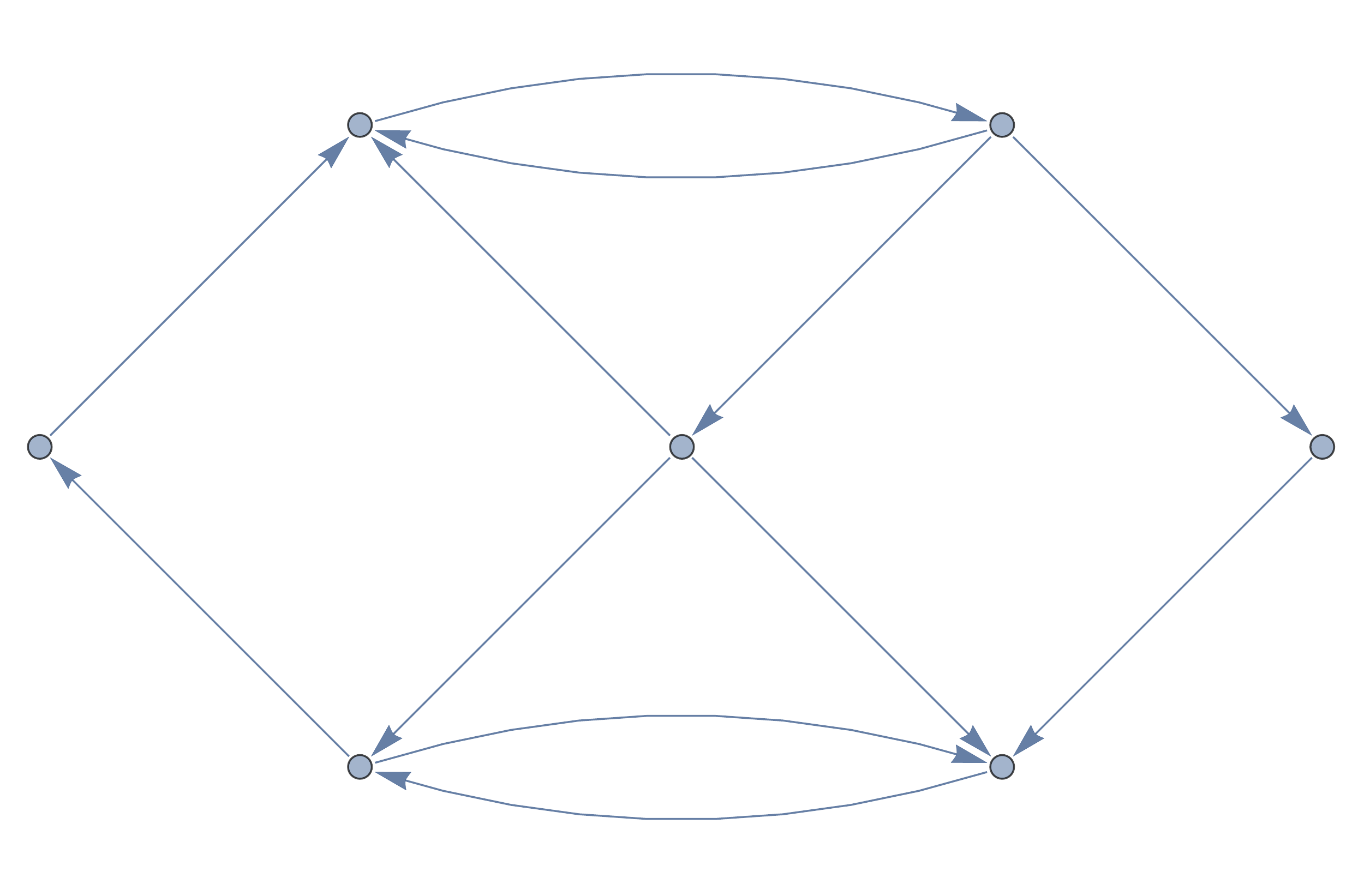}
    \put(47,-5){$G$}
    \put(-5,32){$v_1$}
    \put(22,60){$v_2$}
    \put(73,60){$v_3$}
    \put(99,32){$v_4$}
    \put(48,26.5){$v_5$}
    \put(73,3){$v_6$}
    \put(22,3){$v_7$}
    \end{overpic} &
    \hspace{0.35in}
    \begin{overpic}[scale=.175]{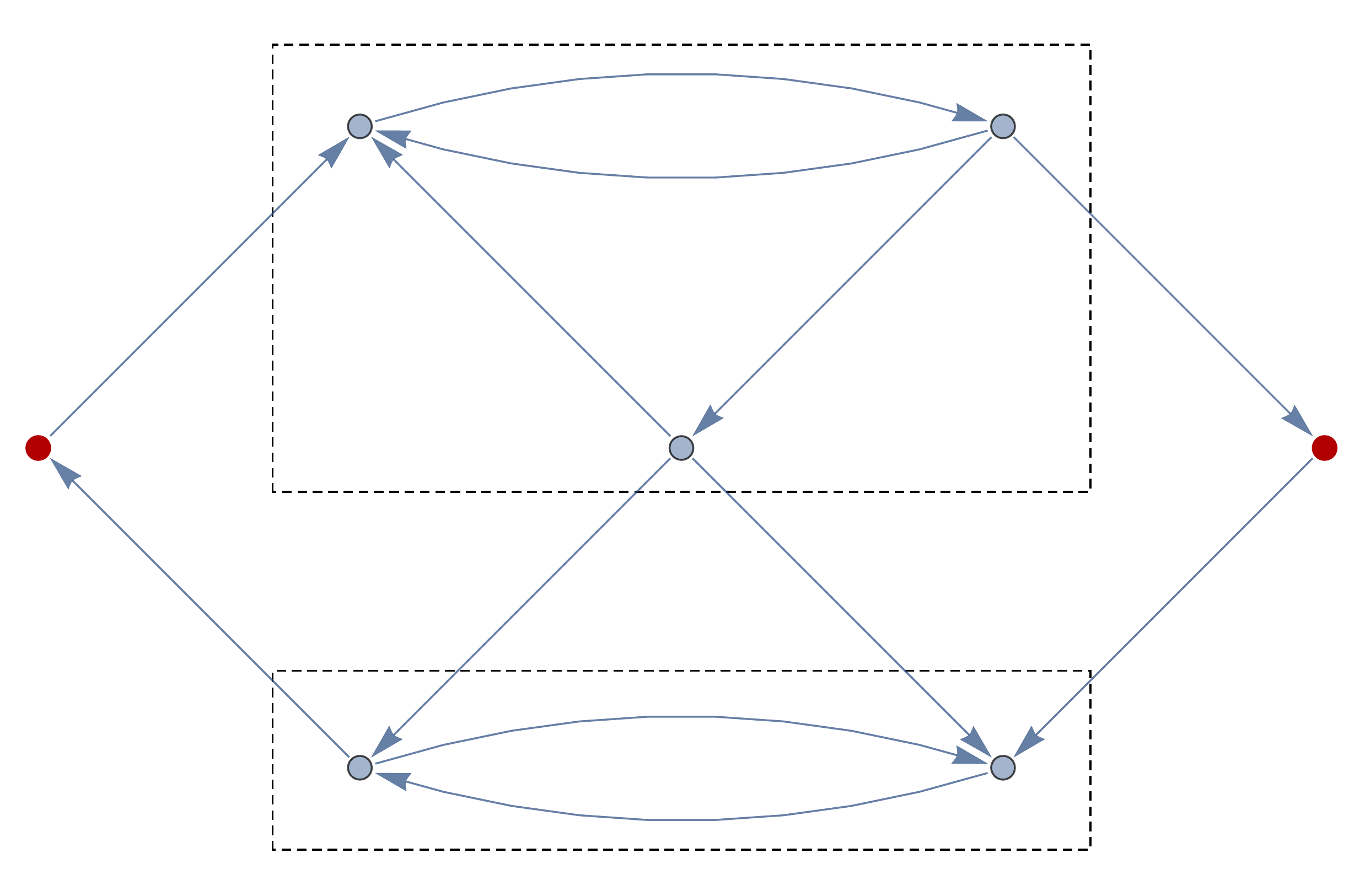}
    \put(46,-5){$G$}
    \put(-5,32){$v_1$}
    \put(99,32){$v_4$}
    \put(46,43){$C_1$}
    \put(46,7.25){$C_2$}
    \end{overpic}
\end{tabular}
\end{center}
  \caption{The unweighted graph $G=(V,E)$ is shown left. The components $C_1$ and $C_2$ of $G$ with respect to the vertex set $S=\{v_1,v_4\}$ are shown right. These components are the subgraphs $C_1=G|\{v_2,v_3,v_5\}$ and $C_2=G|\{v_6,v_7\}$, indicated by the dashed boxes, which are the strongly connected components of the graph $G|{\bar{S}}$.}\label{fig1}
\end{figure}
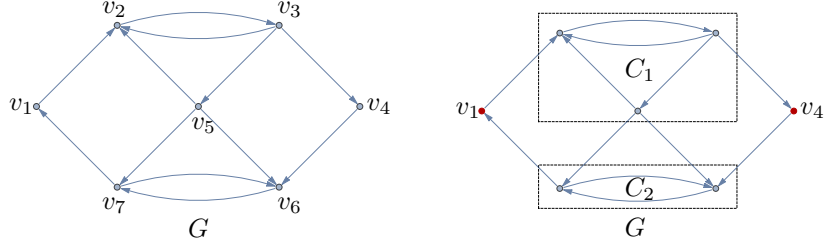

A decomposition of a graph into its component branches is given in the following example.

\begin{example}\label{ex:1} \textbf{(Branch Decomposition)}
For the unweighted graph $G=(V,E)$ shown in figure \ref{fig1} (left) let $S=\{v_1,v_4\}$. Then the graph $G|\bar{S}$ has the strongly connected components $C_1=G|\{v_2,v_3,v_5\}$ and $C_2=G|\{v_6,v_7\}$, which are indicated in figure \ref{fig1} (right). The set $\mathcal{B}_S(G)$ consists of the component branches
\begin{align*}
\beta_1&=v_1,e_{12},C_1,e_{34},v_4 \ \ \ \ \ \ \ \ \ \ \ \ \ \beta_2=v_4,e_{46},C_2,e_{71},v_1\\
\beta_3&=v_1,e_{12},C_1,e_{56},C_2,e_{71},v_1 \ \ \ \beta_4=v_1,e_{12},C_1,e_{57},C_2,e_{71},v_1;
\end{align*}
which are shown in figure \ref{fig2} (left).
\end{example}

It is worth emphasizing that each branch $\beta\in\mathcal{B}_S(G)$ is a subgraph of $G$. As a consequence, the edges of $\beta$ inherit the weights they had in $G$ if $G$ is weighted. If $G$ is unweighted then its component branches are likewise unweighted (cf. figure \ref{fig2}).

Once a graph has been broken into its various branches the idea is to use these branches to construct a new graph that has, at least locally, the same structure as the original graph. More precisely, this new graph will have the same set of components $C_1,\dots,C_m$ as the original graph but the connections between these components will be different. This evolved graph is formed by merging the branches $\mathcal{B}_S(G)$ of the original graph into a new larger graph as follows.

\begin{definition} \textbf{(Evolved Graphs)}\label{def:exp}
 Suppose $G=(V,E,\omega)$ and $S\subseteq V$. Let $\mathcal{X}_S(G)=(\mathcal{V},\mathcal{E},\mu)$ be the \emph{evolved graph} which consists of the component branches $$\mathcal{B}_{S}(G)=\{\beta_1,\dots,\beta_{\ell}\}$$
 in which we \emph{merge}, i.e. identify, each vertex $v\in S$ in any branch $\beta_i$ with the same vertex $v$ in any other branch $\beta_j$.
\end{definition}

Note that, in a component branch $\beta\in\mathcal{B}_S(G)$ only the first and last vertices of $\beta$ belong to the set $S$. The evolved graph $\mathcal{X}_S(G)$ is then the collection of branches $\mathcal{B}_S(G)$ in which we identify an endpoint of two branches if they are the same vertex. This is shown in the following example.

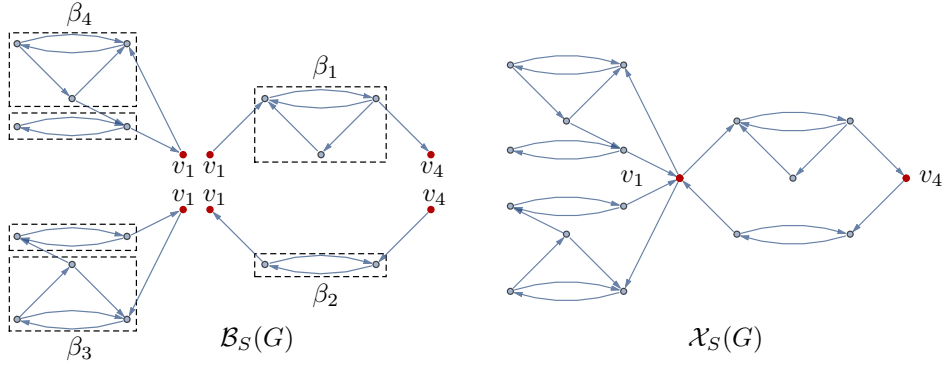
\begin{figure}
\begin{center}
\begin{tabular}{cc}
    \begin{overpic}[scale=.46]{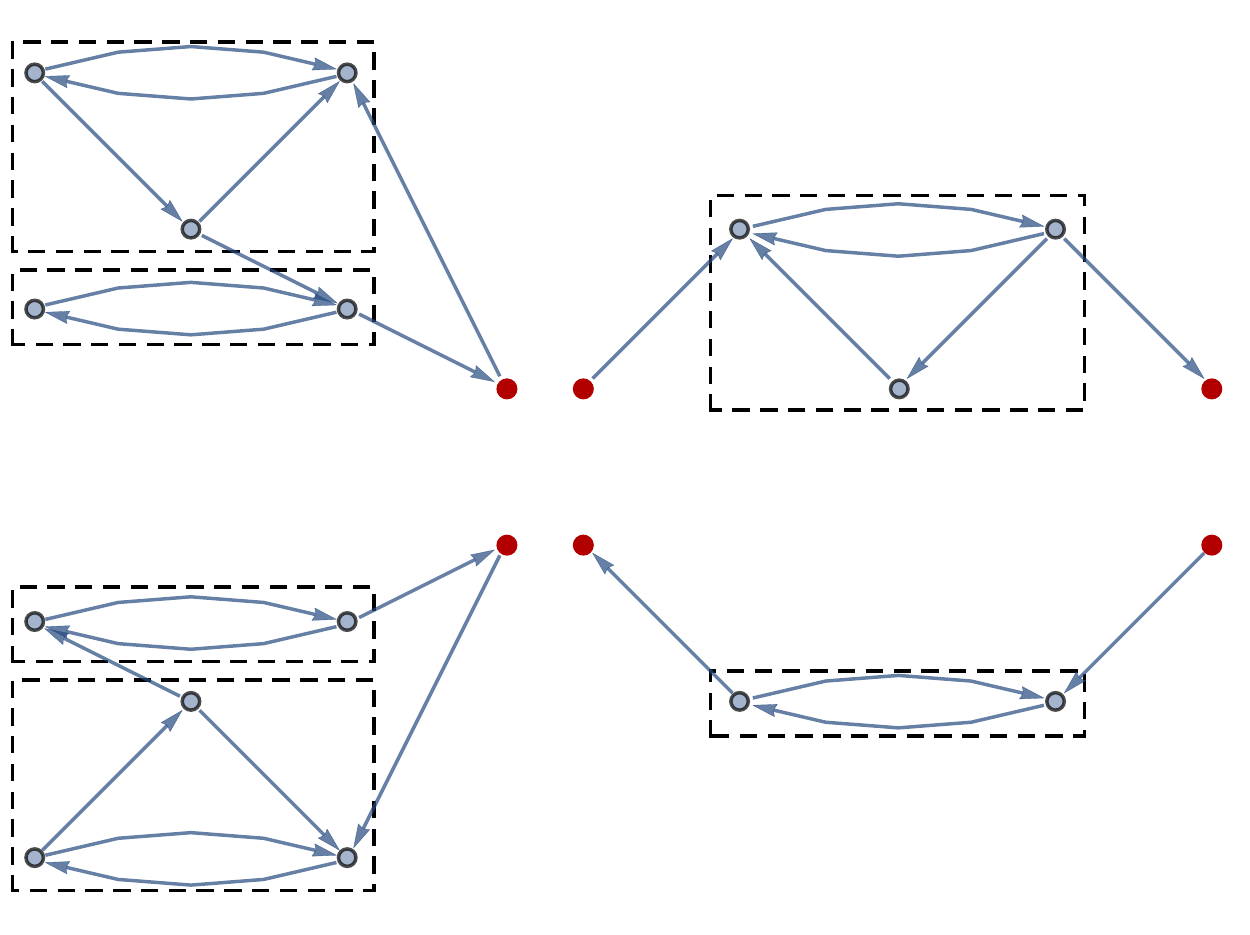}
    \put(49,0){$\mathcal{B}_S(G)$}
    \put(14,74){$\beta_4$}
    \put(14,-2){$\beta_3$}
    \put(70,10){$\beta_2$}
    \put(70,62){$\beta_1$}

    \put(94.5,39.5){$v_4$}
    \put(95,33){$v_4$}
    \put(38,39.5){$v_1$}
    \put(38,33){$v_1$}

    \put(45,39.5){$v_1$}
    \put(45,33){$v_1$}
    \end{overpic} &
    \hspace{0.1in}
    \begin{overpic}[scale=.44]{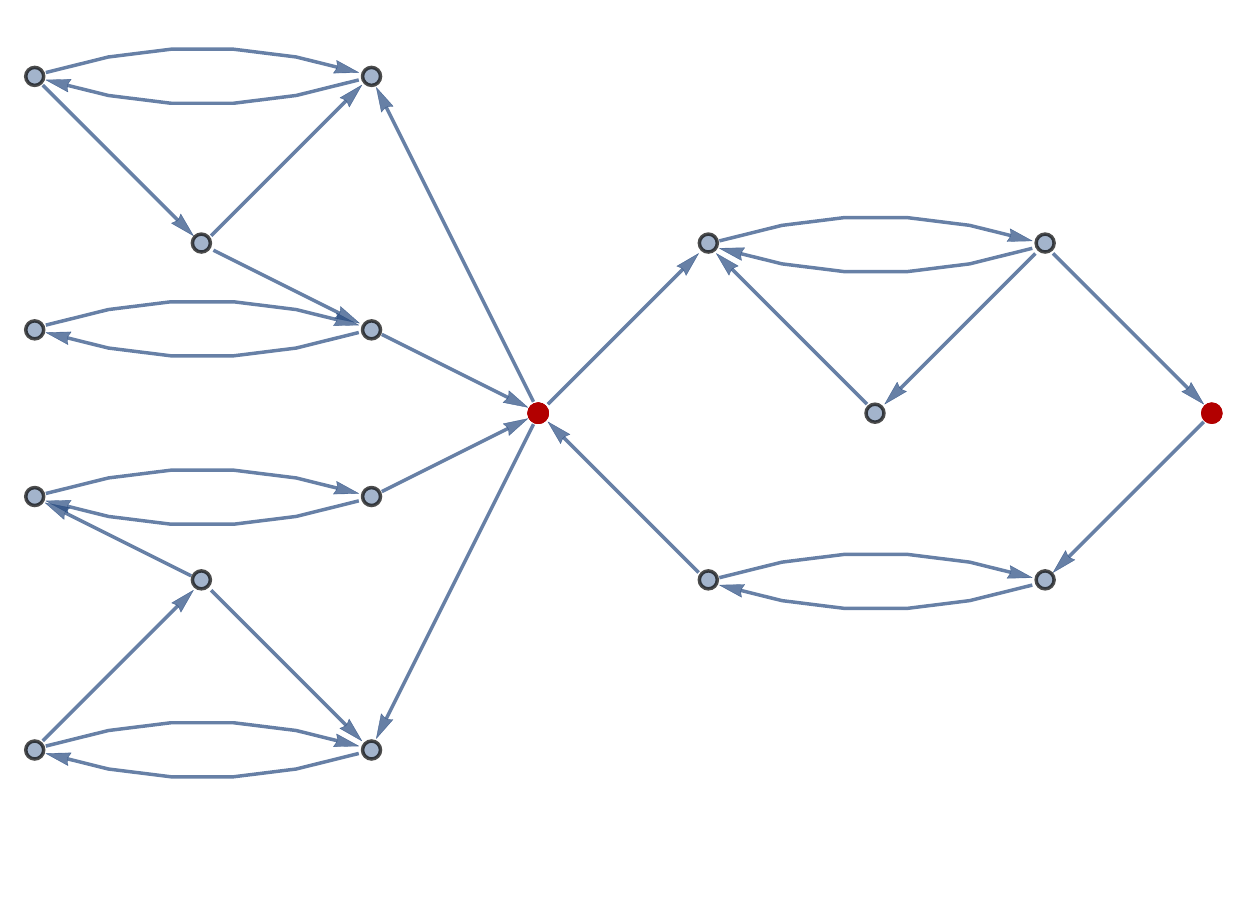}
    \put(45,0){$\mathcal{X}_S(G)$}

    \put(29,39){$v_1$}
    \put(100,39){$v_4$}
    \end{overpic}
\end{tabular}
\end{center}
  \caption{The component branches $\mathcal{B}_S(G)=\{\beta_1,\beta_2,\beta_3,\beta_4\}$ of the graph $G=(V,E)$ from figure \ref{fig1} over the vertex set $S=\{v_1,v_4\}$ are shown (left). The evolved graph $\mathcal{X}_S(G)$ is shown (right), which is made by merging each of the vertices $v_1$ and $v_4$ respectively in each of the branches of $\mathcal{B}_S(G)$. The edge labels and vertex labels are omitted, except for those vertices in $S$, to emphasize how these vertices are merged.}\label{fig2}
\end{figure}

\begin{example}\label{ex:2} \textbf{(Merging Component Branches)}
For the unweighted graph $G=(V,E)$ shown in figure \ref{fig1} (left) we again let $S=\{v_1,v_4\}$. The component branches $\mathcal{B}_S(G)=\{\beta_1,\beta_2,\beta_3,\beta_4\}$ are the branches shown in figure \ref{fig2} (left). By merging each of the vertices $v_1\in S$ over all branches in $\mathcal{B}_S(G)$ and doing the same for the vertex $v_4\in S$ the result is the graph $\mathcal{X}_S(G)=(\mathcal{V},\mathcal{E})$ shown in figure \ref{fig2} (right).
\end{example}

One can think of a graph evolution as a graph transformation that evolves the structure of the graph by maintaining the interactions between the vertices of $S$ and of $\bar{S}$ but breaking
up the interactions that pass from one set to the other. Because of this property, this transformation maintains the local structures, i.e. components of the graph, but reorganizes how these components interact.

In fact, the evolved graph will have more of these components than the original graph it is evolved from. In this sense one can think of the evolved graph as have a more modular structure where the network's modules or communities are formed by these components, which may be highly connected internally but are only minimally connected to the rest of the network (see \cite{Newman2006} for a survey of modularity). Moreover, because there are potentially many copies of the same component in the evolved graph the evolved graph has a certain amount of redundancy, which is a feature that is often observed in real networks \cite{MSA08}.

An important aspect of a graph evolution is that it preserves a graph's weight set. For instance, if $G$ has real, integer-valued, or positive weights then any one of its evolutions will have weights that are real, integer-valued, or positive, respectively. In fact, if $\mathcal{X}_S(G)=(\mathcal{V},\mathcal{E},\mu)$ is an evolution of $G=(V,E,\omega)$ then $\mu(\mathcal{E})=\omega(E)$ so that the edge weights of the evolved graph are collectively the same as the collective edge weights of the original unexpanded graph.

An example of an evolution of a graph with integer weights is shown in figure \ref{fig4} (left). Here the graph $H=(V,E,\omega)$ has the weight set $\omega(E)=\{1,2,3,4\}$. The evolution $\mathcal{X}_S(H)=(\mathcal{V},\mathcal{E},\mu)$ over $S=\{v_1,v_4\}$ is shown in figure \ref{fig4} (right) which, as can be seen, has the same weight set.

To understand how a change in a network's graph structure effects the network's dynamics and in turn the network's function, we need some notion that relates both structure and dynamics. One of the most useful concepts that does this is the notion of a network's spectrum. The spectrum of a network can be defined in a number of ways since there are a number of ways that a matrix can be associated with a network.

Matrices that are often associated with a network include various Laplacian matrices, e.g. the regular Laplacian, combinatorial Laplacian, normalized Laplacian, signless Laplacian, etc. Other matrices include the adjacency and weighted adjacency matrix of a graph, the distance matrix of a graph, etc.
The eigenvalues of these matrices are of interest for a number of reasons. For instance, the \emph{spectral gap} of the Laplacian matrix of a graph, which is its second smallest eigenvalue, determines a number of dynamic properties including synchronization thresholds and the rate of convergence to synchronization and consensus \cite{WM10} on certain networks. The spectral radius of a weighted adjacency matrix of a graph is related to the dynamic stability of the associated network \cite{BW12,BW13}.

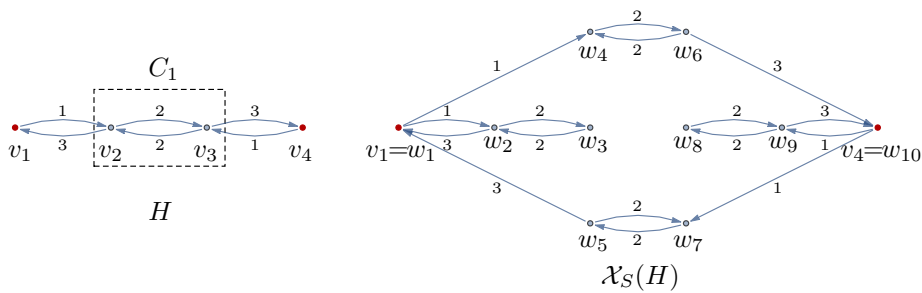
\begin{figure}
\begin{center}
\begin{tabular}{c}
    \begin{overpic}[scale=0.43]{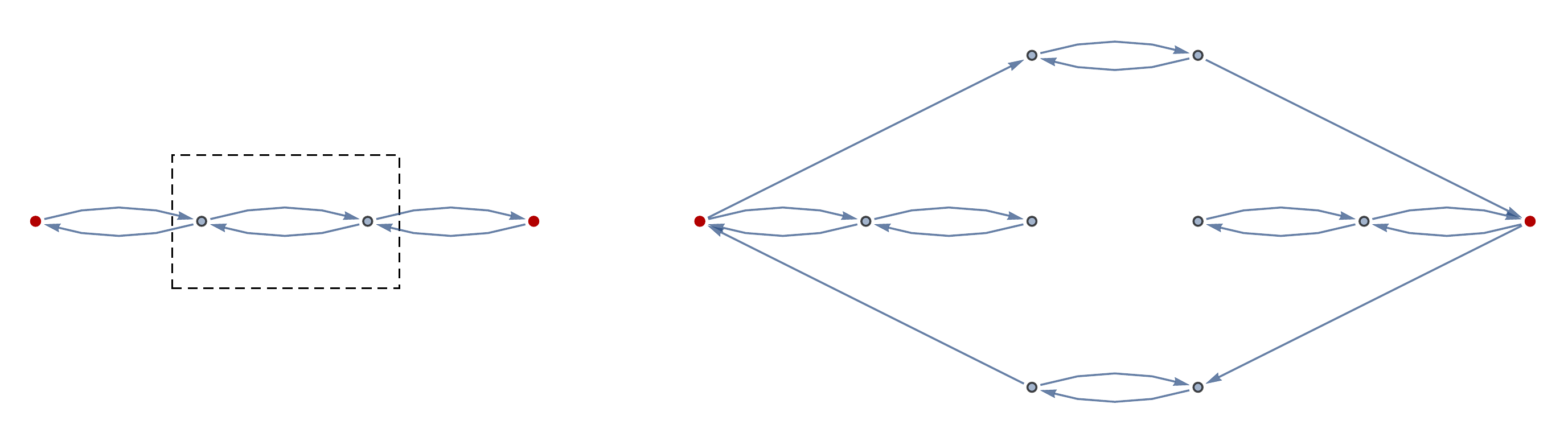}
    \put(17,4){$H$}
    \put(1.5,11){$v_1$}
    \put(11.5,11){$v_2$}
    \put(22,11){$v_3$}
    \put(32.5,11){$v_4$}

    \put(7,15.5){\tiny $1$}
    \put(17.75,15.5){\tiny $2$}
    \put(28.25,15.5){\tiny $3$}
    \put(7,11.75){\tiny $3$}
    \put(17.75,11.75){\tiny $2$}
    \put(28.25,11.75){\tiny $1$}
    \put(17,20){$C_1$}

    \put(41,11){$v_1 \hspace{-0.05in}=\hspace{-0.05in}w_1$}
    \put(54,12){$w_2$}
    \put(64.5,12){$w_3$}
    \put(64.5,22){$w_4$}
    \put(64.5,1){$w_5$}
    \put(75,22){$w_6$}
    \put(75,1){$w_7$}
    \put(75,12){$w_8$}
    \put(85.5,12){$w_9$}
    \put(93.5,11){$v_4 \hspace{-0.05in}= \hspace{-0.05in}w_{10}$}

    \put(49.5,15.5){\tiny $1$}
    \put(59.75,15.5){\tiny $2$}
    \put(49.5,11.75){\tiny $3$}
    \put(59.75,11.75){\tiny $2$}

    \put(81.5,15.5){\tiny $2$}
    \put(91.25,15.5){\tiny $3$}
    \put(81.5,11.75){\tiny $2$}
    \put(91.25,11.75){\tiny $1$}

    \put(70.5,26){\tiny $2$}
    \put(70.5,22.25){\tiny $2$}

    \put(70.5,5){\tiny $2$}
    \put(70.5,1.25){\tiny $2$}

    \put(55,20.5){\tiny $1$}
    \put(55,7){\tiny $3$}

    \put(86,20.5){\tiny $3$}
    \put(86,7){\tiny $1$}

    \put(67,-3){$\mathcal{X}_S(H)$}

    \end{overpic}
\end{tabular}
\end{center}
  \caption{The weighted graph $H$ shown left has the strongly connected component $C_1=H|\{v_2,v_3\}$ with respect to the vertex set $S=\{v_1,v_4\}$. The evolution $\mathcal{X}_S(H)$ is the weighted graph shown right, which has the same weight set as $H$.}\label{fig4}
\end{figure}

The type of matrix we consider in this paper is the weighted adjacency matrix of a graph. Given a graph $G=(V,E,\omega)$ its \emph{weighted adjacency matrix} $M=M(G)$ is the matrix
\[
M_{ij}=
\begin{cases}
\omega(e_{ij}) \ \ \text{if} \ e_{ij}\in E\\
0 \hspace{.8cm} \ \text{otherwise}.
\end{cases}
\]
If $G$ is unweighted then each entry $M(G)_{ij}$ is either zero or one. The \emph{eigenvalues} of the matrix $M(G)$ make up the graph's \emph{spectrum}, which we denote by
\[
\sigma(G)=\{\lambda\in\mathbb{C}:\det(M(G)-\lambda I)=0\}.
\]
In later sections, we will investigate the connection between the spectrum of a graph $G$ and the dynamics of the network associated with it. For now we simply assume that to each network there is an associated graph $G$ with adjacency matrix $M=M(G)$.

The reason we consider the adjacency matrix of a graph $G$ verses any one of the other matrices that can be associated with $G$ is that there is a one-to-one relationship between the matrices $M\in\mathbb{R}^{n\times n}$ and the weighted directed graphs we consider. Hence, we can talk about a unique graph associated with any square matrix with real entries.

Because we are concerned with the spectrum of a graph, which is a set that includes multiplicities, the following will be important for our discussion. First, the element $\alpha$ of the set $A$ that includes multiplicities has \emph{multiplicity} $m$ if there are $m$ elements of $A$ equal to $\alpha$. If $\alpha\in A$ with multiplicity $m$ and $\alpha\in B$ with multiplicity $n$ then\\
\indent (i) the \emph{union} $A\cup B$ is the set in which $\alpha$ has multiplicity $m+n$; and\\
\indent (ii) the \emph{difference} $A-B$ is the set in which $\alpha$ has multiplicity $m-n$ if $m-n>0$ and where $\alpha\notin A-B$ otherwise.

For ease of notation, if $A$ and $B$ are sets that include multiplicity then we let $B^k=\cup_{i=1}^kB$ for $k\geq 1$. That is, the set $B^k$ is $k$ copies of the set $B$ where we let $B^0=\emptyset$. For $k=-1$ we let $A\cup B^{-1}=A-B$. With this notation in place, the spectrum of a graph $G$ and the spectrum of $\mathcal{X}_S(G)$ are related by the following result.

\begin{theorem}\label{thm1} \textbf{(Spectra of Evolved Graphs)}
Let $G=(V,E,\omega)$, $S\subseteq V$, and let $C_1,\dots,C_m$ be the strongly connected components of $G|\bar{S}$. Then
\[
\sigma\big(\mathcal{X}_S(G)\big)=\sigma(G)\cup\sigma(C_1)^{n_1-1}\cup\sigma(C_2)^{n_2-1}\cup\dots\cup \sigma(C_m)^{n_m-1}
\]
where $n_i$ is the number of branches in $\mathcal{B}_S(G)$ that contain $C_i$.
\end{theorem}

If a network has a changing graph structure that can be modeled via a graph evolution, or more naturally a sequence of evolutions, then theorem \ref{thm1} allows us to effectively track the changes in the network's spectrum. This will be used in section \ref{sec:4} to describe how a network can evolve structurally while maintaining a particular type of dynamics. Since a network's dynamics is related to its ability to perform specific tasks, this theory of graph evolutions is introduced as a tool to study the interplay of network growth and function.

Because the proof of theorem \ref{thm1} relies on the theory of isospectral graph reductions \cite{BW12,BWBook} we defer it until the Appendix, where the necessary parts of this theory are given. For now, we consider an example of theorem \ref{thm1}. In section \ref{sec3} we consider some consequences and applications of theorem \ref{thm1} to the study of network structure and growth.

\begin{example}\label{ex:3}
The weighted graph $H=(V,E,\omega)$ in figure \ref{fig4} has eigenvalues
\[\sigma(H)=\{\pm 1,\pm 3\}.\]
For $S=\{v_1,v_4\}$ the graph $H|\bar{S}$ has the strongly connected components $C_1$ with eigenvalues $\sigma(C_1)=\{\pm 2\}$. Since four branches of $\mathcal{B}_S(G)$ contain $C_1$, theorem \ref{thm1} implies that
\[
\sigma(\mathcal{X}_S(H))=\sigma(H)\cup\sigma(C_1)^3=\{\pm 1,\pm 2,\pm 2,\pm 2,\pm 3\}.
\]
\end{example}

Not only are the eigenvalues of a graph $G$ preserved as the graph is evolved but also the eigenvectors of $G$ are also preserved to a certain extent. An \emph{eigenvector} $\mathbf{v}$ of a graph $G$ corresponding to the eigenvalue $\lambda$ is a vector such that $M(G)\mathbf{v}=\lambda\mathbf{v}$. If this is the case then we call $(\lambda, \mathbf{v})$ an \emph{eigenpair} of $G$. Moreover, if $S$ is a subset of the vertices of $G$ then we let $\mathbf{v}_S$ denote the eigenvector $\mathbf{v}$ restricted to those entries indexed by $S$.

\begin{prop}\label{prop:0}\textbf{(Eigenvectors of Evolved Graphs)}
Let $(\lambda,\mathbf{v})$ be an eigenpair of the graph $G=(V,E,\omega)$. If $S\subset V$ and $\lambda\notin\sigma(G|\bar{S})$ then there is an eigenpair $(\lambda,\mathbf{w})$ of the expanded graph $\mathcal{X}_S(G)$ such that $\mathbf{w}_S=\mathbf{v}_S$.
\end{prop}

As an example illustrating how restricted eigenvectors are maintained as a graph is evolved, we again consider the graph $H$ and its expansion $\mathcal{X}_S(H)$ shown in figure \ref{fig4} where $S=\{v_1,v_4\}$. Letting the eigenpairs of $H$ be given by $(\lambda_i,\mathbf{v}^i)$ and the corresponding eigenpairs of $\mathcal{X}_S(G)$ by $(\lambda_i,\mathbf{w}^i)$ for $i=1,2,3,4$ we find that
\begin{align*}
\lambda_1=3:& \ \mathbf{v}^1=[5 \ \ 15 \ \ 15 \ \ 5]^T, \ \ \ \ \ \ \mathbf{w}^1=[5 \ \ 9 \ \ 6 \ \ 6 \ \ 9 \ \ 9 \ \ 6 \ \ 6 \ \ 9 \ \ 5]^T\\
\lambda_2=1:& \ \mathbf{v}^2=[-1 \ -1 \ \ 1 \ \ 1]^T, \ \ \ \ \mathbf{w}^2=[-1 \ \ 1 \ \ 2 \ -2 \ \ 1 \ -1 \ \ 2 \ -2 \ -1 \ \ 1]^T\\
\lambda_3=-1:& \ \mathbf{v}^3=[1 \ -1 \ -1 \ \ 1]^T, \ \ \ \ \mathbf{w}^3=[1 \ \ 1 \ -2 \ -2 \ \ 1 \ \ 1 \ -2 \ -2 \ \ 1 \ \ 1]^T\\
\lambda_4=-3:& \ \mathbf{v}^4=[-5 \ \ 15 \ -15 \ \ 5]^T, \ \mathbf{w}^4=[-5 \ \ 9 \ -6 \ \ 6 \ \ 9 \ -9 \ -6 \ \ 6 \ -9 \ \ 5]^T.
\end{align*}
Note that $\mathbf{v}^i_S=\mathbf{w}^i_S$ for each $i=1,2,3,4$. That is, $\mathbf{v}^1_S=\mathbf{w}^1_S=[5 \ \ 5]^T$, $\mathbf{v}^2_S=\mathbf{w}^2_S=[-1 \ \ 1]^T$, $\mathbf{v}^3_S=\mathbf{w}^3_S=[-1 \ \ 1]^T$, and $\mathbf{v}^4_S=\mathbf{w}^4_S=[-5 \ \ 5]^T$.

Proposition \ref{prop:0} states that the graphs $G$ and $\mathcal{X}_S(G)$ have the same eigenvectors if we restrict our attention to those entries that correspond to $S$ and to those eigenvectors with corresponding eigenvalues in $\sigma(G)-\sigma(G|\bar{S})\subset\sigma(\mathcal{X}_S(G))$. One consequence of this fact is that the eigenvector centrality of the vertices in $S$ remain the same as the graph is expanded.

By the Perron-Frobenius theorem, if $G=(V,E)$ is strongly connected then $G$ has a unique eigenvalue $\rho$ which is the \emph{spectral radius} of $G$, i.e. $\rho=\max\{|\lambda|:\lambda\in\sigma(G)\}$. Moreover, $\rho$ is a simple eigenvalue and the eigenvector $\mathbf{p}$ associated with $\rho$ has nonnegative entries. The vector $\mathbf{p}$, which is unique up to a constant, gives the relative score $p_i$ to each vertex $v_i\in V$. This value $p_i$ is referred to as the \emph{eigenvector centrality} of the vertex $v_i$ \cite{Newman??}. Here we refer to the vector $\mathbf{p}$ as an \emph{eigencentrality vector} of the graph $G$.

The reason the eigencentrality vector $\mathbf{p}$ gives each vertex of $V$ a \emph{relative} score is that for any $c>0$, $c\mathbf{p}$ is also an eigenvector associated with $\rho$. That is, the eigenvector centrality of a vertex is not fixed. However, any vector $\mathbf{p}$ associated with $\rho$ induces the same \emph{ranking} on the vertices of $G$ where $v_i$ is ranked above $v_j$ if $p_i>p_j$, below $v_j$ if $p_i<p_j$, and the same as $p_j$ if $p_i=p_j$. A graph evolution of $G$ preserves its vertices' eigenvector centrality in the following way.

\begin{theorem}\label{prop10}\textbf{(Eigenvector Centrality of Evolved Graphs)}
Let $G=(V,E)$ be strongly connected with eigencentrality vector $\mathbf{p}$. If $S\subset V$ then $\mathcal{X}_S(G)$ has an eigencentrality vector $\mathbf{q}$ where $\mathbf{p}_S=\mathbf{q}_S$. Hence, the eigenvalue centrality of the vertices in $S$ is preserved as the graph is evolved.
\end{theorem}

The graph $H$ and its evolution $\mathcal{X}_S(G)$ in figure \ref{fig4} have eigencentrality vectors
\[
\mathbf{v}^1=[5 \ \ 15 \ \ 15 \ \ 5]^T \ \text{and} \ \mathbf{w}^1=[5 \ \ 9 \ \ 6 \ \ 6 \ \ 9 \ \ 9 \ \ 6 \ \ 6 \ \ 9 \ \ 5]^T
\]
respectively, which correspond to the eigenvalue $\rho=3$. Since $H$ is strongly connected it then follows from theorem \ref{prop10} that $\mathbf{v}^1_S=\mathbf{w}^1_S=[5 \ \ 5]$. That is, the vertices $v_1$ and $v_4$ of $S$ have the same eigenvalue centrality in $H$ and the same eigenvalue centrality in $\mathcal{X}_S(H)$. This illustrates how the eigenvector centrality of the core vertices in a graph remain the same relative to each other as the graph is evolved.

As with the proof of theorem \ref{thm1} the proofs of proposition \ref{prop:0} and theorem \ref{prop10} are given in the Appendix as they relies on the theory of isospectral graph reductions.

\section{Structural Evolution Rules and Equivalence Classes}\label{sec3}

The notion of a graph evolution is extremely versatile in the sense that a graph $G=(V,E,\omega)$ can be evolved with respect to any subset of its vertex set. Hence, if $|V|=n$ then there are $2^n$ potential evolutions of $G$. The main idea presented in this section is that it is possible to focus on exactly one of these evolutions if we choose a particular rule $\tau$ for selecting vertices from a graph.

An important consequence of choosing such a rule $\tau$ is that it will give us a way of comparing the topologies of two distinct networks. In particular, any such rule will allow us to partition the set of networks we consider into different equivalence classes, i.e. separate any set networks we are working with into those networks that are similar and dissimilar with respect to the rule $\tau$.

The idea is that it is possible to partition the graphs we consider with respect to any rule that selects a specific set of vertices from any given graph. We refer to $\tau$ as a \emph{structural rule} if it selects a unique subset of vertices from any graph $G$. For instance, $\tau$ could be the rule that selects all vertices with loops or all the vertices that have a certain centrality, etc. (cf. examples \ref{Zachary} and \ref{ex:evoequ}). Not every rule that can be devised will select a unique vertex set of a graph. The simplest example would be the rule that randomly selects a single vertex of a graph. This selection is, of course, nonunique.

For a rule $\tau$ and a graph $G=(V,E,\omega)$, we let $\tau(V)\subseteq V$ denote the set of vertices this rule selects. For simplicity, we denote the evolution of $G$ and the component branches with respect to $\tau$ as
\[
\tau(G)=\mathcal{X}_{\tau(V)}(G) \ \ \text{and} \ \ \mathcal{B}_{\tau}(G)=\mathcal{B}_{\tau(V)}(G),
\]
respectively. It is worth emphasizing that any graph can be \emph{evolved} with respect to any rule $\tau$ that selects a unique subset of its vertex set.

In the following example we consider two different rules and investigate to what extent these rules lead to different evolutions of the same network.

\begin{figure}
\begin{center}
\begin{tabular}{cc}
    \begin{overpic}[scale=.65]{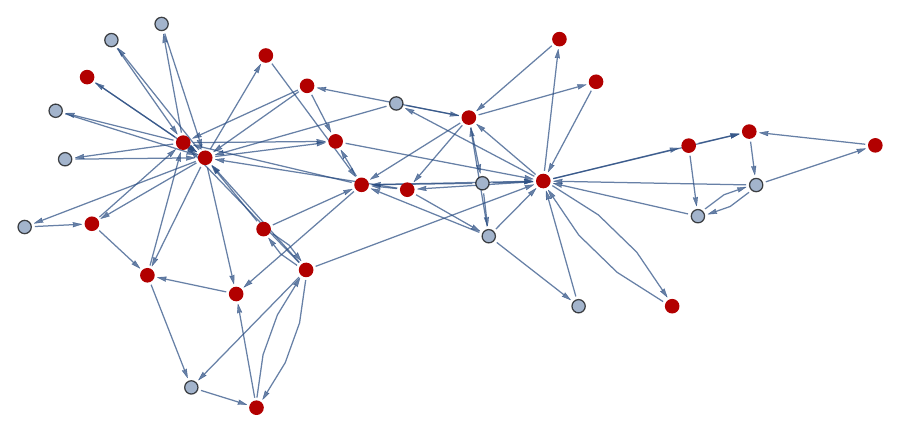}
    \put(38,0){\small $ZK_{EC}$}
    \end{overpic} &
    \begin{overpic}[scale=.65]{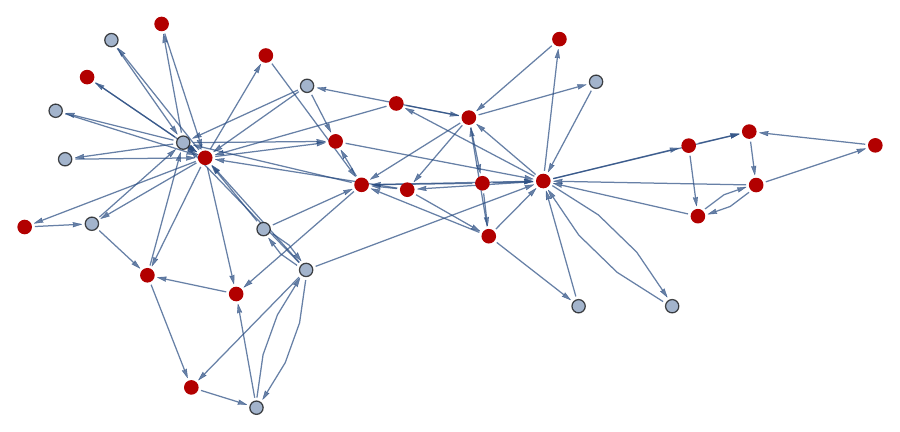}
    \put(38,0){\small $ZK_{BC}$}
    \end{overpic}\\\\
    \begin{overpic}[scale=.65]{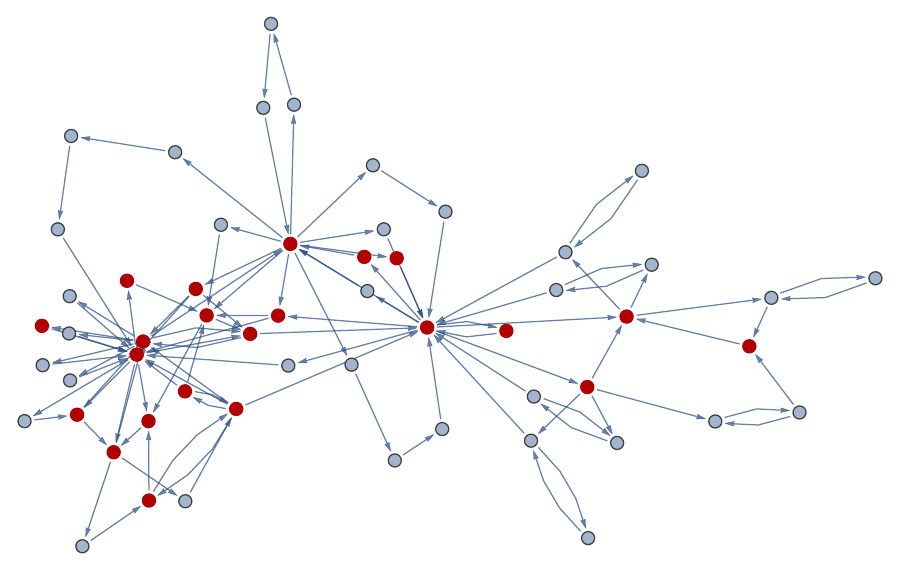}
    \put(31,0){\small $\tau_{EC}(ZK_{EC})$}
    \end{overpic} &
    \begin{overpic}[scale=.575]{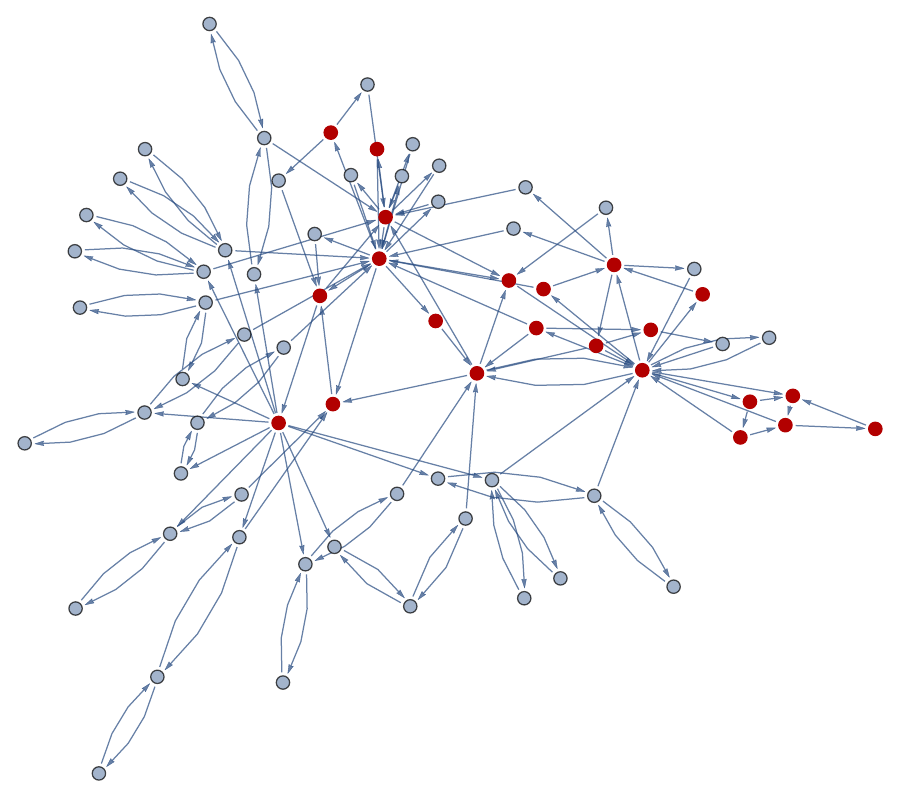}
    \put(31,0){\small $\tau_{BC}(ZK_{BC})$}
    \end{overpic}
\end{tabular}
\end{center}

\vspace{0.1in}

\caption{The graphs $ZK_{EC}$ and $ZK_{BC}$ shown above are identical and are directed versions of the Zachary Karate Club graph. In $ZK_{EC}$ and $ZK_{BC}$ two-thirds of the vertices with the largest eigenvector centrality and betweenness centrality are highlighted, respectively. The same vertices are highlighted in the evolved graphs $\tau_{EC}(ZK_{EC})$ and $\tau_{BC}(ZK_{BC})$ of the directed karate club graph over these distinct vertex sets, respectively.}\label{fig:Zach}
\end{figure}

\begin{example}\label{Zachary}\textbf{(Zachary Karate Club)}
We consider a variant of the Zachary Karate Club graph. The Zachary Karate Club graph itself is an undirected graph on thirty-four vertices with seventy-eight edges. The graph describes the social network consisting of thirty-four members of a karate club where the graph's edges represent those members who interact outside the club \cite{Z77}.

Here, for the sake of illustration, we consider a directed version of this graph in which the edges of this original karate club graph are oriented and a number of these directed edges are randomly removed. Since this graph is directed we note that one can interpret this orienting of edges as changing the social interactions in the club from ``friends" to ``followers". The idea is that \emph{friends} have a joint interaction in which there is no distinction in direction. For a \emph{follower}, one person follows another, which is a directed and often non-reciprocated interaction.

The directed version of the Zachary Karate Club graph we consider is shown in figure \ref{fig:Zach} at the top left and right. Here, $ZK=ZK_{EC}=ZK_{BC}$. The difference between $ZK_{EC}$ and $ZK_{BC}$ are the vertices that are highlighted in each. In $ZK_{EC}$ the vertices are highlighted using the rule $\tau_{EC}$ that selects the top two-thirds of the club's members with largest \emph{eigenvector centrality}. The vertices highlighted in $ZK_{BC}$ are those selected by the rule $\tau_{BC}$ that selects the top two-thirds of the club's members with largest \emph{betweenness centrality}.

Here the vertex sets $\tau_{EC}(V)$ and $\tau_{BC}(V)$ are distinct. Consequently, each rule results in a different graph evolution. These different evolutions of $ZK$ are shown in figure \ref{fig:Zach} at the bottom left and right, respectively. One of the most noticeable difference between the two graphs are their sizes. In particular, $|\tau_{EC}(ZK_{EC})|=57$ whereas $|\tau_{BC}(ZK_{})|=72$. On the other hand, both graphs have nearly the same spectrum. Using theorem \ref{thm1} one can quickly compute that
\begin{align*}
\sigma(\tau_{EC}(ZK))&=\sigma(ZK)\cup\{\pm1\}^5\cup\{0\}^{13}\\
\sigma(\tau_{BC}(ZK))&=\sigma(ZK)\cup\{\pm\sqrt{2}\}^{11}\cup\{0\}^{29}.
\end{align*}

In terms of network growth, one can interpret the highlighted club members in $ZK_{EC}$ and $ZK_{BC}$ as the \emph{core} members of the group respectively, who are looking to recruit new members while simultaneously maintaining existing relationships. The network grows in that the core members find new members through the other non-core members of the group. One way to interpret this is to assume that these non-core members are likely to have contacts outside the group who, when they are introduced to the group, end up meeting those people who are contacts of the person who introduced to the club in the first place. If this is the case, the resulting growth of the club can be modeled via a graph evolution.

In this example the core members are determined by either their eigenvector or betweenness centrality, i.e. the rules $\tau_{EC}$ and $\tau_{BC}$, but these are only two of the $2^{34}$ possible ways of modeling the growth of this network. A more accurate model of growth would likely involve designing a rule $\tau$ for selecting core members that incorporates information such as the club's historic growth, city demographics, etc.
\end{example}

In general, a rule for network growth is likely to be highly dependent on the specific network under consideration. The idea presented here is that, if a core set of elements can be identified then the network's growth could be modeled using a graph evolution.

As previously mentioned, another potential use for this theory of network evolution is to introduce a new way of determining whether two networks are similar. In order to determine whether two graph are similar with respect to some rule $\tau$ we need the notion of a graph isomorphism, which can be defined as follows.

Two weighted digraphs $G_1=(V_1,E_1,\omega_1)$ and $G_2=(V_2,E_2,\omega_2)$ are \emph{isomorphic} if there is a bijection $\phi:V_1\rightarrow V_2$ such that there is an edge $e_{ij}$ in $G_1$ from $v_i$ to $v_j$ if and only if there is an edge $\tilde{e}_{ij}$ between $\phi(v_i)$ and $\phi(v_j)$ in $G_2$ with $\omega_2(\tilde{e}_{ij})=\omega_1(e_{ij})$. If the map $\phi$ exists it is called an \emph{isomorphism} and we write $G_1\simeq G_2$. Note that since any weighted/unweighted directed/undirected graph can be considered to be a weighted directed graph, this definition of isomorphic applies to all such graphs.

The equivalent notion for matrices is that the matrix $A\in\mathbb{R}^{n\times n}$ is \emph{similar} to the matrix $B\in\mathbb{R}^{n\times n}$ by some permutation matrix $P$, i.e. $A=PBP^{-1}$. Thus, if two graphs are isomorphic then their spectra are identical. This notion of being isomorphic, together with the fact that a structural rule $\tau$ generates the unique graph $\tau(G)$, allow us to define the following equivalence relation. The idea here is that two graph are similar with respect to a rule $\tau$ if they both evolve to the \emph{same}, i.e. isomorphic graph under this rule.

\begin{theorem}\textbf{(Evolution Equivalence)}\label{thm2}
Suppose $\tau$ is a structural rule. Then $\tau$ induces an equivalence relation $\sim$ on the set of all weighted directed graphs where $G\sim H$ if
\[
\tau(G)\simeq\tau(H).
\]
When this holds, we call $G$ and $H$ \emph{evolution equivalent} with respect to $\tau$.
\end{theorem}

\begin{proof}
For any $G=(V,E,\omega)$ the set $\tau(V)\subseteq V$ is unique implying the graph $\tau(G)=\mathcal{X}_{\tau(V)}(G)$ is uniquely determined by the rule $\tau$. Clearly, the relation of being evolution equivalent with respect to $\tau$ is reflexive and symmetric, i.e. $\tau(G)\simeq\tau(H)$ and $\tau(G)\simeq\tau(H)$ imply $\tau(G)\simeq\tau(G)$ and $\tau(H)\simeq\tau(G)$, respectively. Also, if $\tau(G)\simeq\tau(H)$ and $\tau(H)\simeq\tau(K)$ then $\tau(G)\simeq\tau(K)$ completing the proof.

\end{proof}

Theorem \ref{thm2} states that every structural rule $\tau$ can be used to divide the set of graphs we consider, and by association all networks, into subsets. These subsets, or more formally \emph{equivalence classes}, are those graphs that share a common topology with respect to $\tau$. By \emph{common topology} we mean that graphs in the same class have the same set of component branches and therefore evolve into the same graph under $\tau$.

One reason for studying these equivalence classes is that it may not be obvious and most often is not that two different graphs belong to the same class. That is, two graphs may be structurally similar but until both graphs are evolved this may be difficult to see. One of the ideas we introduce here is that by choosing an appropriate rule $\tau$ one can discover this similarity. This is demonstrated in the following example.

\begin{figure}
\begin{center}
\begin{tabular}{ccc}
    \begin{overpic}[scale=.33,angle=-180.5]{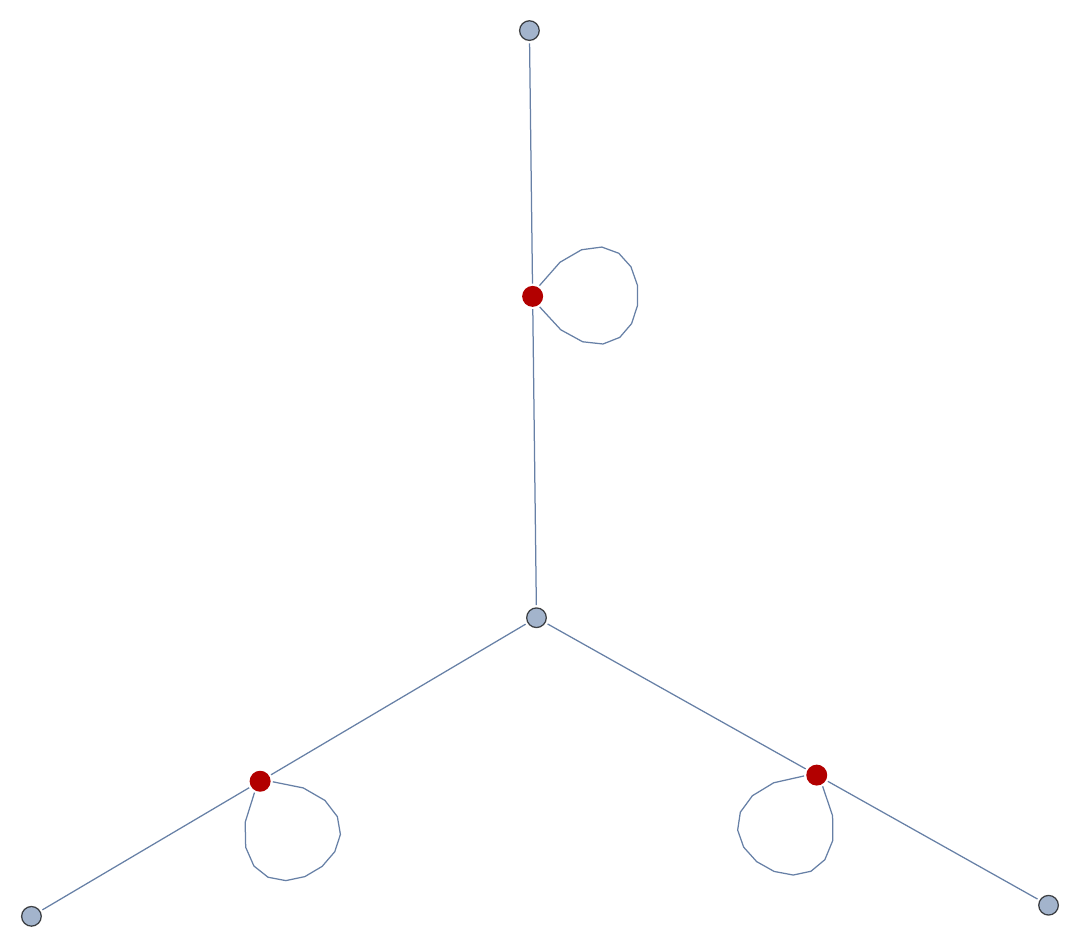}
    \put(38,-5){\put(10,-2){$G$}}
    \end{overpic} &
    \begin{overpic}[scale=.42,angle=181]{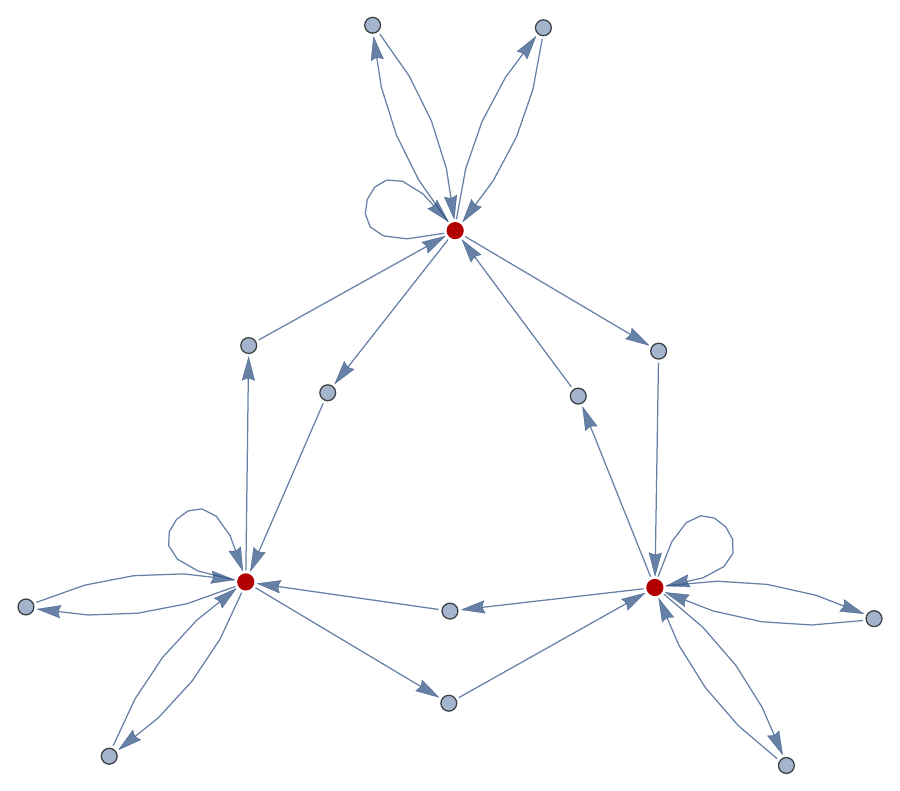}
    \put(25,-7){$\ell(G)\simeq\ell(H)$}
    \end{overpic} &
    \begin{overpic}[scale=.26,,angle=91]{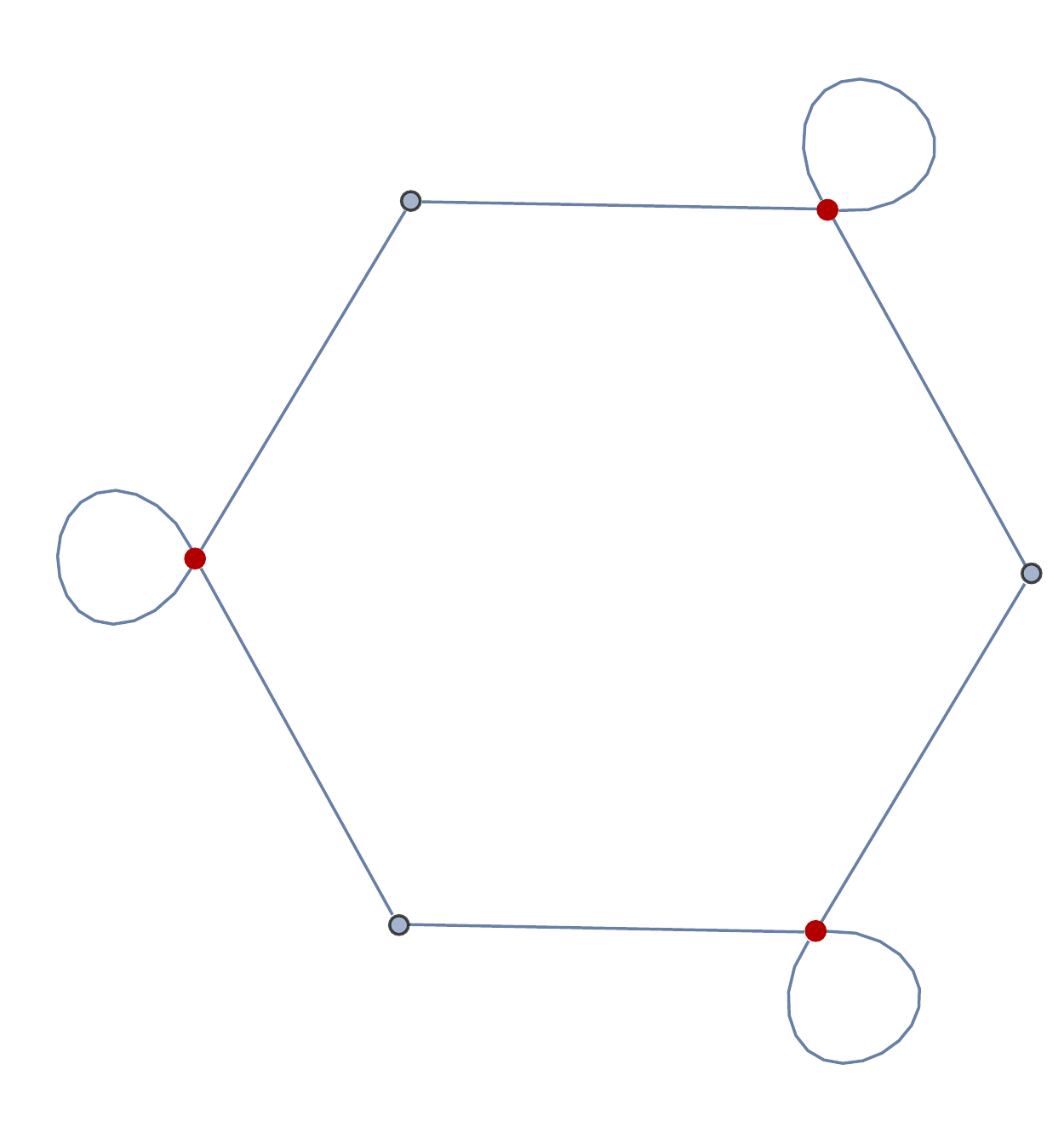}
    \put(46,-7){$H$}
    \end{overpic}
\end{tabular}
\end{center}
  \caption{The graph $G$ and the graph $H$ are evolution equivalent with respect to the rule $\ell$ that selects those vertices of a graph that have loops. That is, the graphs $\ell(G)$ and $\ell(H)$ are isomorphic as is shown.}\label{fig6}
\end{figure}

\begin{example}\label{ex:evoequ} \textbf{(Evolution Equivalent Graphs)}
Consider the unweighted undirected graphs $G=(V_1,E_1)$ and $H=(V_2,E_2)$ shown in figure \ref{fig6}. Without their loops, $G$ is an extended 3-star graph and $H$ is the 6-cycle graph. Here, we let $\ell$ be the rule that selects all vertices of a graph that have loops. We let $S_1=\ell(V_1)$ and $S_2=\ell(V_2)$ be these vertices of $G$ and $H$ selected by $\ell$ respectively, which are the vertices highlighted in figure \ref{fig6} in $G$ and $H$, respectively.

Although $G$ and $H$ appear to be quite different, the graphs $\ell(G)$ and $\ell(H)$ are isomorphic as is shown in figure \ref{fig6} (center). Hence, $G$ and $H$ belong to the same family of graphs with respect to the structural rule $\ell$. Moreover, based on theorem \ref{thm1} and the fact that the strongly connected components of both $G|\bar{S}_1$ and $H|\bar{S}_2$ consist of a collection of disconnected vertices then
\[
\sigma(\ell(G))=\sigma(G)\cup\{0\}^8 \ \ \text{and} \ \ \sigma(\ell(H))=\sigma(H)\cup\{0\}^9.
\]
Hence, $G$ and $H$ have the same nonzero spectrum.

It is worth noting that two graphs can be equivalent under one rule but not another. For instance, if $w$ is the structural rule that selects vertices without loops then $w(G)\not\simeq w(H)$ although $\ell(G)\simeq \ell(H)$.
\end{example}

From a practical point of view, a rule $\tau$ allows those studying a particular class of networks a way of comparing the \emph{evolved topology} of these networks and drawing conclusions about both the evolved and original networks. Of course, the rule $\tau$ should be designed by the particular biologist, chemist, physicist, etc. to have some significance with respect to the networks under consideration.

In fact, one can drop the notion of a fixed rule completely and evolve different networks over different subsets that are deemed important to the respective biologist, chemist, or physicist. Although having no fixed rule means that we no longer have the equivalence relation guaranteed by theorem \ref{thm2}, it is still possible and potentially much more useful to compare the evolved topology of different networks which have been evolved in different ways. The point is, many networks currently under study are likely to have similar features that come to light as these networks topology is evolved. A simple example is the following.

\begin{example}\label{ex:semiequ} \textbf{(Nonequivalent Graphs)}
Consider the graphs $K$ and $L$ shown in figure \ref{fig:nonequiv}. These graphs are locally indistinguishable from one another in that they have the same degree distributions. To determine to what extent $K$ and $L$ have a similar global structure with respect to a certain type of growth we need a rule that describes the core vertices of a network that the graph can evolve around.

For example, letting $\eta$ be the rule that selects the vertices of the graph of smallest degree we find that the branch sets $\mathcal{B}_{\eta}(K)$ and $\mathcal{B}_{\eta}(L)$ are nearly identical (see figure \ref{fig:nonequiv} bottom, left and right). The difference is that the components $C_1$ and $C_2$ in these branches are not identical. Hence, $K$ and $L$ are not evolution equivalent under $\eta$. However, given that $C_1$ and $C_2$ are fairly similar, it is not surprising that the spectra
\begin{align*}
\sigma(G)&=\{2.813,-1.342,0.529,-1,-1,-0.414,-0.347,-2\}\\
\sigma(H)&=\{\pm 2.813,\pm 1.342,\pm 0.529,\pm 1\}
\end{align*}
share many of the same values.
\end{example}

In general there are many ways of comparing how similar the two branch sets $\mathcal{B}_{\tau}(G)$ and $\mathcal{B}_{\tau}(H)$ are. For instance, one could simply compare the number of branches or compare the branches without their weights. Any such method would be a way of measuring how similar any two graphs are with respect to $\tau$.

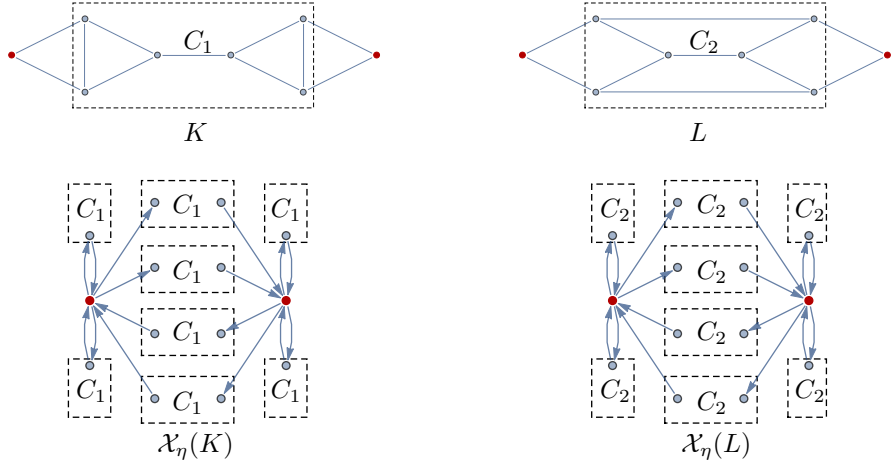
\begin{figure}
\begin{center}
\begin{tabular}{c}
    \begin{overpic}[scale=.35]{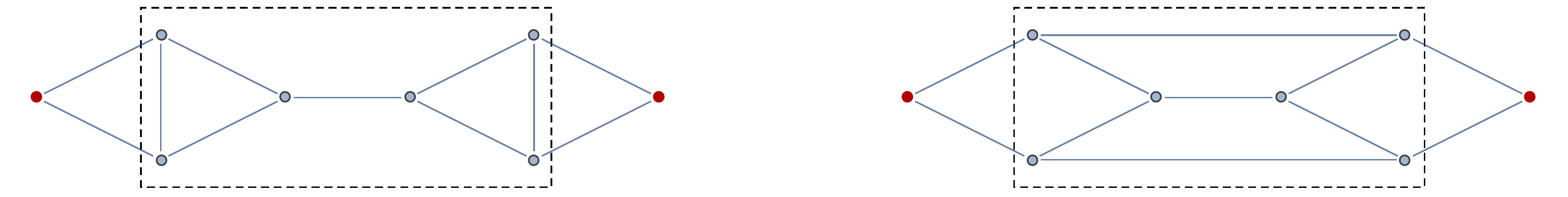}
    \put(21,-3){$K$}
    \put(76,-3){$L$}
    \put(21,7){$C_1$}
    \put(76,7){$C_2$}
    \end{overpic}\\\\\\
    \begin{overpic}[scale=.5]{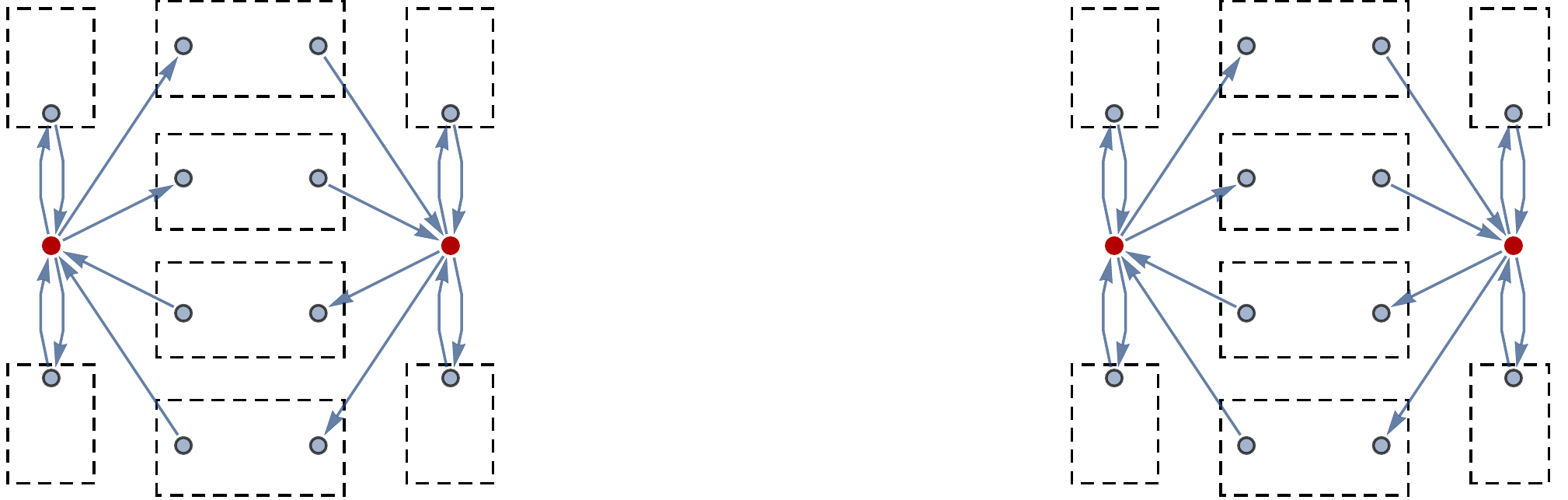}
    \put(1.5,3.5){$C_1$}
    \put(1.5,27){$C_1$}
    \put(14,2){$C_1$}
    \put(14,11){$C_1$}
    \put(14,19){$C_1$}
    \put(14,27.5){$C_1$}
    \put(26.75,3.5){$C_1$}
    \put(26.75,27){$C_1$}
    \put(12,-3.5){$\mathcal{X}_\eta(K)$}

    \put(69.5,3.5){$C_2$}
    \put(69.5,27){$C_2$}
    \put(82,2){$C_2$}
    \put(82,11){$C_2$}
    \put(82,19){$C_2$}
    \put(82,27.5){$C_2$}
    \put(94.75,3.5){$C_2$}
    \put(94.75,27){$C_2$}
    \put(80,-3.5){$\mathcal{X}_\eta(L)$}
    \end{overpic}
\end{tabular}
\end{center}

\vspace{0.1in}

\caption{The graphs $K$ and $L$ (above) have a similar structure of branches (below) with respect to the rule $\eta$ that selects those vertices of a graph with smallest degree. The difference is the component $C_1$ of $K$ is not the same as the component $C_2$ of $L$, so that $G$ and $H$ are not evolution equivalent with respect to $\eta$.}\label{fig:nonequiv}
\end{figure}

Not only can a rule $\tau$ be used to discover the similarities between two graphs or networks but this rule can be used to \emph{sequentially evolve} the topology of the graph. Inductively, we let $\tau^n(G)$ denote the evolution of $\tau^{n-1}(G)$ with respect to $\tau$, so that $\tau^n(G)=\tau(\tau^{n-1}(G))$ where $\tau^0(G)=G$. The result is the sequence of graphs
\[
G, \ \tau(G), \ \tau^2(G), \ \tau^3(G),\dots
\]
generated from the graph $G$ by the rule $\tau$.

The connection to real-world networks is the idea that, under certain conditions the network $G$ will evolve over time according to some fixed rule $\tau$. This notion of network evolution is illustrated in the following example.

\begin{example}\label{ex0} \textbf{(Sequential Graph Evolutions)}
Suppose $G=(V,E,\omega)$. For $v\in V$ let $d_{in}(v)$ be the \emph{in-degree} of $v$, which is the number of incoming edges incident to $v$, excluding loops in $G$. In an undirected graph the in-degree of a vertex $v$ is the same as the number of non-loop edges incident to $v$. Here we let $\delta$ be the rule
\begin{equation}\label{eq:deg}
\delta(V):=\{v\in V: d_{in}(v)=1\}.
\end{equation}
Observe that for any graph $G$ the set $\delta(V)$ both exists and is unique. Here, the vertex set $\delta(V)\subseteq V$ of the graph $G=(V,E)$ in figure \ref{fig:exfractal} (left) is highlighted. The result of sequentially expanding $G$ using $\delta$ is shown for the first few iterates of $\delta$ in figure \ref{fig:exfractal}.

It is worth mentioning that it is possible to show that
\[
\sigma(\delta^k(G))=\sigma(G)\cup\{0\}^{n_k} \ \ \text{for each} \ \ k\geq 0,
\]
where $n_k=2^4(2^k-1)$. We note that although the nonzero spectrum of the graph does not change as the graph evolves, the topology of the graph $\tau^n(G)$ becomes increasingly complicated as $n\rightarrow\infty$.
\end{example}

Example \ref{ex0} raises a few natural questions regarding the sequential evolution of a graph, or more generally, the sequential evolution of different classes of graphs over certain types of rules. One is, given a graph $G$ and a rule $\tau$, what is the spectrum of $\sigma(\tau^n(G))$ as $n$ increases. Another related question is, what is the topology of the graph $\tau^n(G)$ as $n$ increases. That is, what is the graph's \emph{asymptotic spectrum} and \emph{topology} under $\tau$?

Since the notion of a graph evolution is a new direction in the study of network dynamics, these are currently open questions. For the second question regarding the graph's asymptotic structure, we note that the graph $G$ in figure \ref{fig:exfractal} experiences exponential growth as it evolves. In contrast, the graphs $G$ and $H$ in figure \ref{fig6} stop evolving after one iterate as $\ell^2(G)=\ell(G)$ and $\ell^2(H)=\ell(H)$. This suggests that the evolution of a graph's topology strongly depends on the particular graph and structural rule $\tau$ used to evolve the its topology.

One can also generalize the notion of sequentially evolving a graph $G$ over a single rule to $\tau$ to sequentially evolving $G$ over the sequence of rules $\tau_1$, $\tau_2$, $\tau_3,\dots$ where each $\tau_i$ is a structural rule. The result is the sequence of graphs
\[
G, \ \tau_1(G), \ \tau_2(\tau_1(G)), \ \tau_3(\tau_2(\tau_1(G))),\dots
\]
That is, a network's evolution may more naturally be modeled not by a single rule but some number of rules. For instance, a graph may evolve under a rule $\tau_1$ until it reaches a particular size and from this point evolve under some other rule $\tau_2$ because of restrictions imposed by its environment, age, etc. In this more generally setting, one can similarly investigate the graph's asymptotic spectrum and structure as it evolves under this sequence of rules.

The notion of sequentially evolving a graph also allows us to extend the notion of evolution equivalence given in theorem \ref{thm2}. For a given structural rule $\tau$ we say that the graphs $G$ and $H$ are \emph{weakly evolution equivalent} under $\tau$ if there are numbers $p$ and $q$ such that $\tau^p(G)\simeq\tau^q(H)$.

\begin{figure}
\begin{center}
\begin{tabular}{c}
    \begin{overpic}[scale=.28]{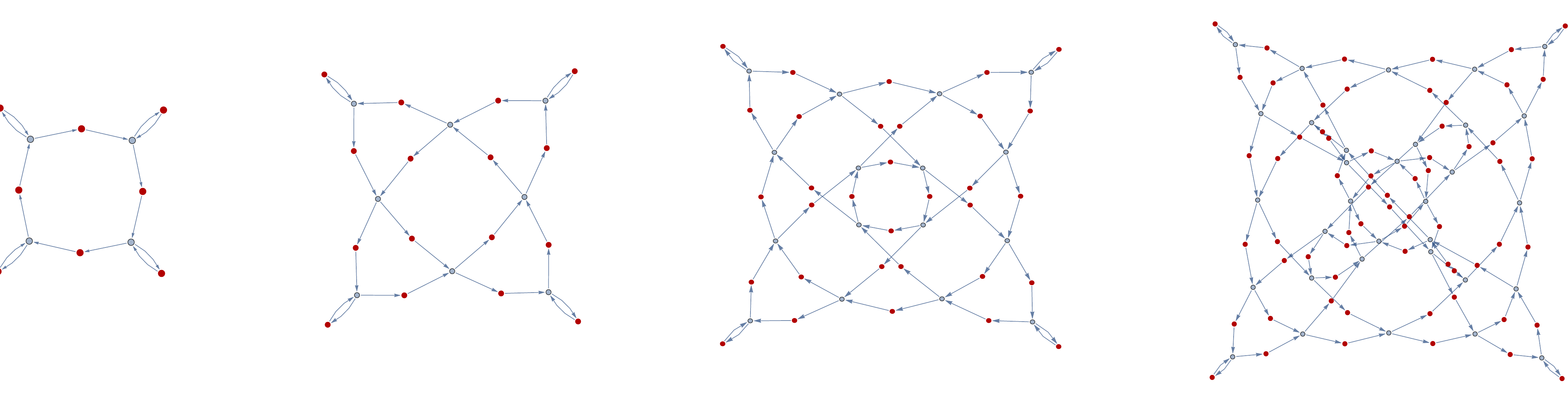}
    \put(4,-2){\small $G$}
    \put(26,-2){\small $\delta(G)$}
    \put(53,-2){\small $\delta^2(G)$}
    \put(85,-2){\small $\delta^3(G)$}
    \end{overpic}
\end{tabular}
\end{center}

\vspace{0.1in}

\caption{The graph $G$ is sequentially evolved using the rule $\delta$ that selects all vertices with in-degree one. The vertices satisfying this property are highlighted red in each iterate of G.}\label{fig:exfractal}
\end{figure}
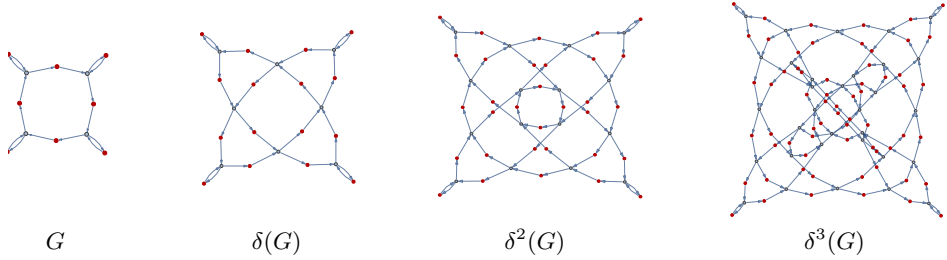

Similar to the proof of theorem \ref{thm2}, one can show that any structural rule $\tau$ induces an equivalence relation $\sim$ on the set of all weighted directed graphs where $G\sim H$ if $G$ and $H$ are weakly evolution equivalent. Hence, every structural rule $\tau$ can be used to partition the space of graphs we consider into subsets that are weakly evolution equivalent under $\tau$.

It is worth noting that if $p=q=1$ where $\tau^p(G)\simeq\tau^q(H)$ then $G$ and $H$ are not only weakly evolution equivalent but also \emph{evolution equivalent}. Since any two graphs $G$ and $H$ that are evolution equivalent under $\tau$ are also weakly evolution equivalent under $\tau$, the evolution equivalent graphs form a subset of those graphs that are weakly evolution equivalent. Hence, the partition induced by those graphs that are evolution equivalent is a \emph{refinement} of the partition induced by those graphs that are only weakly evolution equivalent under $\tau$.

Before ending this section we note that the theory of graph evolutions presented here applies equally well to matrices. The reason is that given a matrix $A\in\mathbb{R}^{n\times n}$ there is a unique graph $G$ for which $M(G)=A$. Hence, for a given evolution rule $\tau$ we can define $$\tau(A)=M(\tau(G)) \ \ \text{where} \ \ A=M(G)$$
to be the \emph{matrix evolution} of $A$ with respect to $\tau$.

This notion of a matrix evolution will be used in the following section to describe how a network can evolve structurally while maintaining its dynamics.

\section{Network Growth and Stability}\label{sec:4}

In sections \ref{sec2} and \ref{sec3} of this paper we have been primarily concerned with the evolution of a network's topology. This evolution determines to a certain extent the network's function and how well the network preforms this function. However, the network's performance also depends on the type of dynamics that emerges from the interactions between the network elements.

One of the more complicated processes to model is the growth of a network that needs to maintain a specific type of dynamics. Some of the most natural examples come from the biological sciences. As previously mentioned, the network of cells in a beating heart attempts to maintain this function even as the network grows. Similarly, in the technological sciences electrical grids are designed to carry power to the consumer even as new lines, plants, etc. are added to the grid.

In general, the \emph{dynamics} of a network with a fixed structure of interactions can be modeled by iterating a map $F:X\rightarrow X$ on a product space $X=\bigoplus_{i\in N} X_{i}$ where $N=\{1,\dots,n\}$ and each local phase space $(X_i,d)$ is a metric space. Here the dynamics of the $i$th network element is given by the $i$th component of $F$:
\[
F_i:\bigoplus_{j\in I_i} X_j\rightarrow X_i, \hspace{0.01in} I_i\subseteq N,
\]
where the set $I_i$ indexes the elements that \emph{interact} with the $i$th element. We refer to the system $(F,X)$ as a \emph{dynamical network}.

The dynamics of the network $(F,X)$ is generated by iterating the function $F$ such that if $\mathbf{x}(k)\in X$ is the state of the network at time $k\geq 0$ then $\mathbf{x}(k+1)=F(\mathbf{x}(k))$ is the state of the network at time $k+1$. The $i$th component $x_i(k)$ represents the state of the $i$th network element at time $k+1$.

The specific type of dynamics we consider here is global stability, which is observed in a number of important systems including neural networks \cite{Cao2003,Cheng2006,SChena2009,MCohen1983,LTao2011}, in epidemic models \cite{Wang2008}, and is also important in the study of congestion in computer networks \cite{Alpcan2005}. In a \emph{globally stable} network, which we will simply refer to as \emph{stable}, the state of the network tends towards equilibrium irrespective of its present state. That is, there is an $\bar{\mathbf{x}}\in X$ such that for any $\mathbf{x}(0)\in X$, $\mathbf{x}(k)\rightarrow\bar{\mathbf{x}}$ as $k\rightarrow\infty$.

This globally attracting equilibrium is typically a state in which the network can carry out a specific task. Whether or not this equilibrium stays stable depends on a number of factors including external influences such as changes in the environment the network is in. Not only can outside influences destabilize a network but potentially the network's own growth. For instance, in a biological networks cancer is the abnormal growth of cells that can lead to network failure. Here we propose a general mechanism describing how a network can evolve structurally while remaining stable.

For simplicity in our discussion we assume that the map $F:X\rightarrow X$ is differentiable and that each $X_i$ is some closed interval of real numbers. Under this assumption we define the following matrix which can be used to investigate the stability of a given dynamical network.

\begin{definition} \textbf{(Stability Matrix)}
For the dynamical network $(F,X)$ suppose there exist finite constants
\begin{equation}\label{eq:stability}
\Lambda_{ij}=\sup_{\mathbf{x}\in X}\Big|\frac{\partial F_i}{\partial x_j}(\mathbf{x})\Big| \ \text{for} \ 1\leq i,j\leq n.
\end{equation}
Then we call the matrix $\Lambda\in\mathbb{R}^{n\times n}$ the \emph{stability matrix} of $(F,X)$.
\end{definition}

The stability matrix $\Lambda$ can be thought of as a global linearization of the typically nonlinear dynamical network $(F,X)$. In fact, the matrix $\Lambda$ can be used to describe the topology of $(F,X)$ in that the graph we associate with $(F,X)$ is the graph $G$ where $M(G)=\Lambda$.

The original motivation for defining $\Lambda$ is that if the dynamical network $(\Lambda, X)$ given by $\mathbf{x}(k+1)=\Lambda\mathbf{x}(k)$ is stable then the same is true of the original network $(F,X)$. To make this precise we let $\rho(\Lambda)$ denote the \emph{spectral radius} of $\Lambda$, i.e.
\[
\rho(\Lambda)=max_i\{|\lambda_i| : \lambda_i\in\sigma(\Lambda)\}.
\]
The fact that stability of the linearized network $(\Lambda,X)$ implies stability of the original network $(F,X)$ is summarized in the following result, the proof of which can be found in \cite{BW12}.

\begin{theorem}\label{stability} \textbf{(Network Stability)}
Suppose $\Lambda$ is the stability matrix of the dynamical network $(F,X)$. If $\rho(\Lambda)<1$ then the dynamical network $(F,X)$ is stable.
\end{theorem}

An important aspect of the dynamic stability described in theorem \ref{stability} is that it is not the standard notion of stability. In \cite{BW13} it is shown that if $\rho(\Lambda)<1$ then the dynamical network $(F,X)$ is not only stable but remains stable even if time-delays are introduced into the network's interactions. Since the introduction of such time-delays can have a destabilizing effect on a network, the type of stability considered in theorem \ref{stability} is stronger than the standard notion of stability. To distinguish between these two types of stability, the stability described in theorem \ref{stability} is given the following name \cite{BW13}.

\begin{definition}\label{def:intrinsic} \textbf{(Intrinsic Stability)}
If $\rho(\Lambda)<1$, where $\Lambda$ is a stability matrix of the dynamical network $(F,X)$, then we say that this network is \emph{intrinsically stable}.
\end{definition}

From an applications point of view, a system that is intrinsically stable is more robust with respect to changes in its environment that cause time delays. Since time delays are unavoidable in most any real network this suggests that, if feasible, it is preferable to design a stable network that is intrinsically stable versus one that is only stable. The goal in this section is to describe how intrinsic stability is also a natural notion for stability of a network with an evolving topology. The idea is that, not only is it possible to evolve the graph structure $G=(V,E,\omega)$ of a network with respect to an evolution rule $\tau$ but it is also possible to evolve the structure of a dynamical network $(F,X)$ with respect to this rule and maintain its stability under certain conditions.

Consider the class of dynamical networks $(F,X)$ having components of the form
\begin{equation}\label{eq:netclass}
F_i(\mathbf{x})=\sum_{j=1}^n A_{ij}f_{ij}(x_j), \ \ i=1,\dots,n
\end{equation}
where the \emph{interaction matrix} $A\in\{0,1\}^{n\times n}$ is an $n\times n$ matrix of zeros and ones and $f_{ij}:X_j\rightarrow\mathbb{R}$ are functions with bounded derivatives for all $1\leq i,j\leq n$. It is worth noting that the matrix $A$ in equation \eqref{eq:netclass} could be absorbed into the functions $f_{ij}$. The reason we use $A$ is for convenience as it will be the means by which we will evolve $(F,X)$.

\begin{definition}\textbf{(Topological Evolutions of Dynamical Networks)}
Let $\tau$ be a structural rule. Then the \emph{evolution} of $(F,X)$ in \eqref{eq:netclass} with respect to $\tau$ is the dynamical network $(F_{\tau},X_\tau)$ with components
\[
(F_{\tau}(\mathbf{x}))_i=\sum_{j=1}^m \tau(A)_{ij}f_{\tau(ij)}(x_j), \ \ i=1,\dots,m
\]
where $X_\tau=\mathbb{R}^{m}$ for $\tau(A)\in\mathbb{R}^{m\times m}$. Here $\tau(ij)=pq$ is the index such that $A_{\tau(ij)}=\tau(A)_{pq}$. We let $\Lambda_\tau$ denote the stability matrix of the evolved network $(F_{\tau},X_{\tau})$.
\end{definition}

To give an example of an evolving network we consider the class of dynamical networks known as discrete-time recurrent neural networks (DRNN). The stability of such systems has been the focus of a large number of investigations, especially time-delayed versions of these systems \cite{LWSL07}. The class of DRNN $(R,X)$ we consider has the form
\begin{equation}\label{eq:DRNN}
R_i(\mathbf{x})=a_i\mathbf{x}_i+\sum_{j=1,j\neq i}^n b_{ij}g_j(x_j)+c_i, \ \ i=1,\dots,n.
\end{equation}
Here the component $R_i$ describes the dynamics of the $i$th neuron where the matrix $D=diag(d_1,\dots,d_n)$ with $|d_i|<1$ is the \emph{feedback coefficient matrix}, the matrix $B\in\mathbb{R}^{n\times n}$ with $b_{ii}=0$ is the \emph{connection weight matrix}, and the constants $c_i$ are the \emph{exogenous inputs} to the network. In the general theory of recurrent neural networks the functions $g_j:\mathbb{R}\rightarrow\mathbb{R}$ are typically assumed to be differentiable, monotonically increasing, and bounded. Here, we make the additional assumption that each $g_j$ has a bounded derivative.

Before continuing we note that equation \eqref{eq:DRNN} can be written in the form of equation \eqref{eq:netclass} by setting
\begin{equation}\label{eq:DRNN2}
f_{ij}(x_j)=
\begin{cases}
a_jx_j+c_j, &\text{for} \hspace{0.2cm} i=j\\
b_{ij}g_j(x_j), &\text{for} \hspace{0.2cm} i\neq j
\end{cases}
\hspace{0.2cm} \text{and} \hspace{0.2cm}
A_{ij}=
\begin{cases}
0 \hspace{0.2cm} \text{if} \hspace{0.1cm} f_{ij}(x_j)\equiv0,\\
1 \hspace{0.2cm} \text{otherwise}.
\end{cases}
\end{equation}
In the following example we evolve the structure of a DRNN.

\begin{example}\label{ex:DRNN}\textbf{(Topological Evolution of a DRNN)}
Consider the recurrent dynamical network $(R,\mathbb{R}^4)$ given by \eqref{eq:DRNN} where $a_j=\alpha$, $b_{ij}=\beta$ for $i\neq j$ and is zero otherwise, $c_j=\gamma$, and $g_j(x)=\tanh(x)$ for $1 \leq i,j\leq 4$ where $|\alpha|<1$ and $\beta\in\mathbb{R}$.

Let $A$ be the interaction matrix
\[
A=
\left[\begin{array}{cccc}
1&1&0&0\\
1&1&1&0\\
0&0&1&1\\
1&0&1&1
\end{array}\right].
\]
Here we choose the function $g_j(x)=\tanh(x)$ as this is a standard activation function used to model neural interactions in network science as it converts continuous inputs to binary outputs. The dynamical network $(R,\mathbb{R}^4)$ is then given by
\[
R(\mathbf{x})=
\left[\begin{array}{l}
\alpha x_1+\beta[\tanh(x_2)]+\gamma\\
\alpha x_2+\beta[\tanh(x_1)+\tanh(x_3)]+\gamma\\
\alpha x_3+\beta[\tanh(x_4)]+\gamma\\
\alpha x_4+\beta[\tanh(x_1)+\tanh(x_3)]+\gamma
\end{array}\right].
\]

We evolve the topology of the network $(R,\mathbb{R}^4)$ using the rule $\delta(V)=\{v\in V: d_{in}(v)=1\}$, which is the same rule used in example \ref{ex0}. This rule evolves this matrix $A$ and the dynamical network $(R,\mathbb{R}^4)$ into $\delta(A)$ and $(R_{\delta},\mathbb{R}^6)$ where
\[
\delta(A)=
\left[\begin{array}{cccccc}
1&0&0&1&1&0\\
0&1&1&0&0&1\\
1&0&1&0&0&0\\
0&1&0&1&0&0\\
1&0&0&0&1&0\\
0&1&0&0&0&1
\end{array}\right] \ \text{and} \
R_{\delta}(\mathbf{x})=\left[\begin{array}{l}
\alpha x_1+\beta[\tanh(x_4)+\tanh(x_5)]+\gamma\\
\alpha x_2+\beta[\tanh(x_3)+\tanh(x_6)]+\gamma\\
\alpha x_3+\beta[\tanh(x_1)]+\gamma\\
\alpha x_4+\beta[\tanh(x_2)]+\gamma\\
\alpha x_5+\beta[\tanh(x_1)]+\gamma\\
\alpha x_6+\beta[\tanh(x_2)]+\gamma
\end{array}\right],
\]
respectively. Using the fact that $\sup_{x\in\mathbb{R}}|\frac{d}{dx}\tanh(x)|=1$ the stability matrix $\Lambda$ of $(R,\mathbb{R}^4)$ and the stability matrix $\Lambda_{\delta}$ of $(R_{\delta},\mathbb{R}^6)$ are given by
\[
\Lambda=\left[\begin{array}{cccc}
|\alpha|&|\beta|&0&0\\
|\beta|&|\alpha|&|\beta|&0\\
0&0&|\alpha|&|\beta|\\
|\beta|&0&|\beta|&|\alpha|
\end{array}\right] \ \ \text{and} \ \
\Lambda_{\delta}=
\left[\begin{array}{cccccc}
|\alpha|&0&0&|\beta|&|\beta|&0\\
0&|\alpha|&|\beta|&0&0&|\beta|\\
|\beta|&0&|\alpha|&0&0&0\\
0&|\beta|&0&|\alpha|&0&0\\
|\beta|&0&0&0&|\alpha|&0\\
0&|\beta|&0&0&0&|\alpha|
\end{array}\right].
\]
The graphs $G$ and $\delta(G)$ associated with the stability matrices $\Lambda$ and $\Lambda_{\delta}$ are shown in figure \ref{fig0} left and center, respectively.

Importantly, $\delta(\Lambda)=\Lambda_{\delta}$ so that, if the stability matrix of $(R,\mathbb{R}^4)$ is evolved by $\delta$ the result is the stability matrix of the expanded network $(R_{\delta},\mathbb{R}^6)$.
\end{example}

The fact that $\delta(\Lambda)=\Lambda_{\delta}$ in this example is not a coincidence but is a simply a consequence of how dynamical network evolutions are defined. That is, if any dynamical network $(F,X)$ given by \eqref{eq:netclass} with stability matrix $\Lambda$ is evolved with respect to the structural rule $\tau$ the resulting dynamical network $(F_{\tau},X_{\tau})$ has the stability matrix $\Lambda_{\tau}=\tau(\Lambda)$. It is therefore possible to test the stability of an evolved version of a dynamical network $(F,X)$ by evolving its stability matrix $\Lambda$.

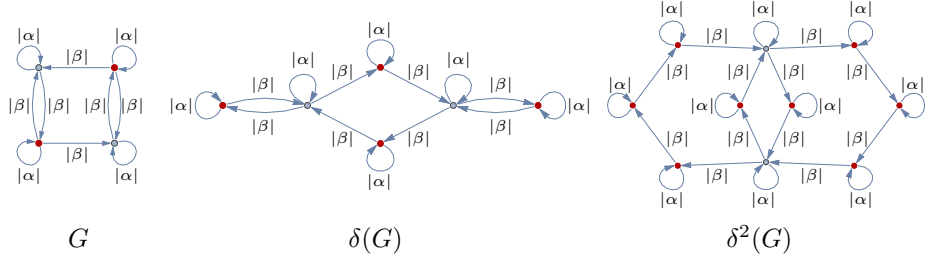
\begin{figure}
\begin{center}
\begin{tabular}{c}
    \begin{overpic}[scale=.48]{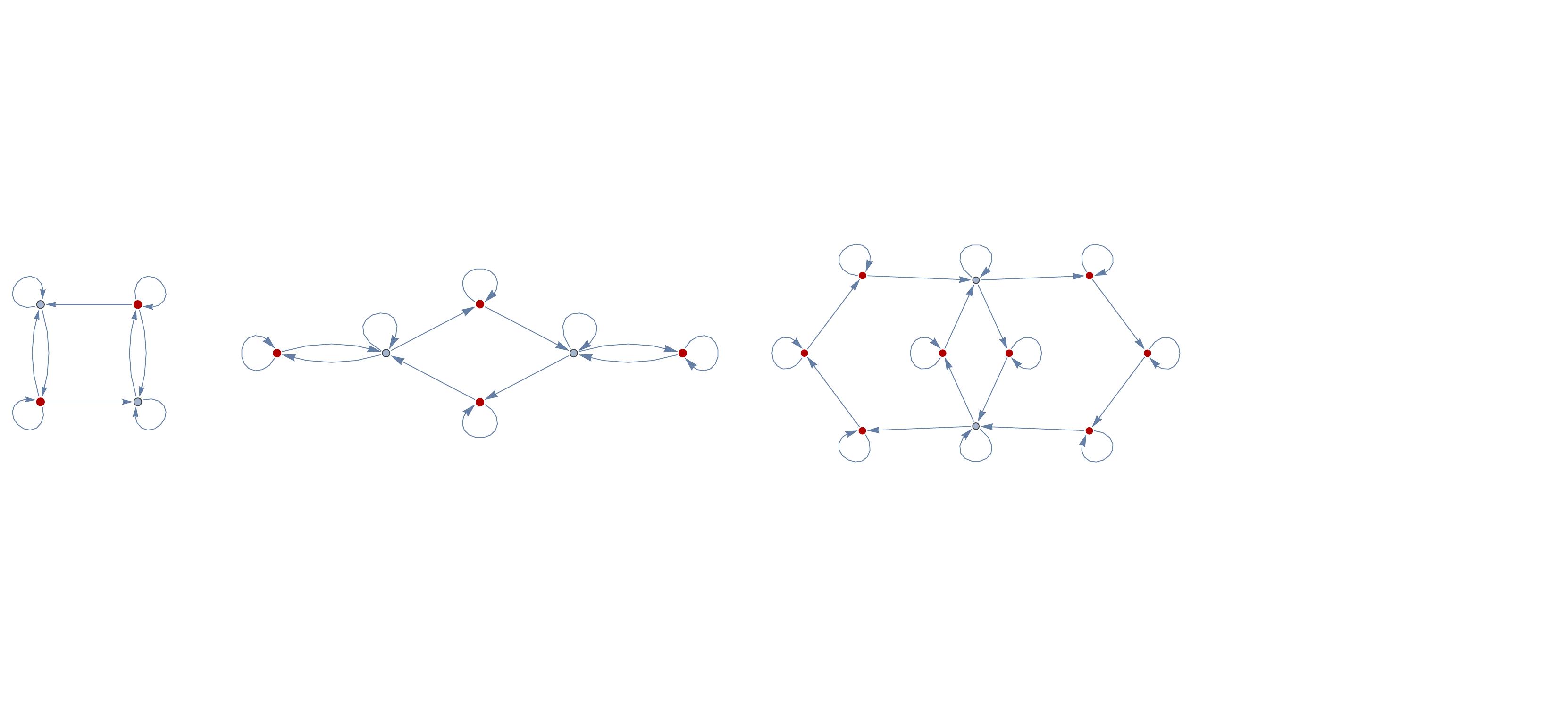}
    \put(5,11){$G$}
    \put(28,11){$\delta(G)$}
    \put(59,11){$\delta^2(G)$}

    \put(.75,16.5){\tiny $|\alpha|$}
    \put(.75,28){\tiny $|\alpha|$}
    \put(8.75,16.5){\tiny $|\alpha|$}
    \put(8.75,28){\tiny $|\alpha|$}
    \put(0,22.25){\tiny $|\beta|$}
    \put(3.25,22.25){\tiny $|\beta|$}
    \put(6.25,22.25){\tiny $|\beta|$}
    \put(9.25,22.25){\tiny $|\beta|$}
    \put(4.75,18.25){\tiny $|\beta|$}
    \put(4.75,26.25){\tiny $|\beta|$}

    \put(13.25,22.25){\tiny $|\alpha|$}
    \put(45.9,22.25){\tiny $|\alpha|$}
    \put(23.25,26){\tiny $|\alpha|$}
    \put(36,26){\tiny $|\alpha|$}
    \put(29.5,28.5){\tiny $|\alpha|$}
    \put(29.5,16){\tiny $|\alpha|$}
    \put(20,24){\tiny $|\beta|$}
    \put(20,20.5){\tiny $|\beta|$}
    \put(39.5,24){\tiny $|\beta|$}
    \put(39.5,20.5){\tiny $|\beta|$}
    \put(26.5,25){\tiny $|\beta|$}
    \put(26.5,19.5){\tiny $|\beta|$}
    \put(32.75,25){\tiny $|\beta|$}
    \put(32.75,19.5){\tiny $|\beta|$}

    \put(49.25,24){\tiny $|\alpha|$}
    \put(73.25,24){\tiny $|\alpha|$}
    \put(56,22.25){\tiny $|\alpha|$}
    \put(66.5,22.25){\tiny $|\alpha|$}

    \put(53.5,30.25){\tiny $|\alpha|$}
    \put(61,30.25){\tiny $|\alpha|$}
    \put(69,30.25){\tiny $|\alpha|$}
    \put(53.5,14.5){\tiny $|\alpha|$}
    \put(61,14.5){\tiny $|\alpha|$}
    \put(69,14.5){\tiny $|\alpha|$}

    \put(57.25,16.5){\tiny $|\beta|$}
    \put(57.25,28){\tiny $|\beta|$}
    \put(65,16.5){\tiny $|\beta|$}
    \put(65,28){\tiny $|\beta|$}

    \put(54,19.5){\tiny $|\beta|$}
    \put(54,25){\tiny $|\beta|$}
    \put(68.75,19.5){\tiny $|\beta|$}
    \put(68.75,25){\tiny $|\beta|$}

    \put(59,19.5){\tiny $|\beta|$}
    \put(59,25){\tiny $|\beta|$}
    \put(63.55,19.5){\tiny $|\beta|$}
    \put(63.55,25){\tiny $|\beta|$}
    \end{overpic}
\end{tabular}
\end{center}

\vspace{-.8in}

\caption{The sequence $G$, $\delta(G)$, and $\delta^2(G)$ represents the topology of the recurrent dynamical network $(R,\mathbb{R}^4)$ and its sequence of evolutions $(R_{\delta},\mathbb{R}^6)$, and $(R_{\delta^2},\mathbb{R}^{10})$ respectively, considered in example \ref{ex:DRNN}. The vertices selected by the rule $\delta$ are highlighted in each graph. The parameters $\alpha,\beta\in\mathbb{R}$.}\label{fig0}
\end{figure}

In previous studies of dynamical networks including DRNN a goal has been to determine under what condition(s) a given network has stable dynamics (see for instance the references in \cite{LWSL07}). Here we consider a different but related question which is, under what condition(s) does a dynamical network with an evolving structure of interactions maintain its stability as it evolves. As a partial answer to this quite general question, we show that if the dynamical network is not only stable but intrinsically stable then its stability is preserved under the evolution of its topology with respect to any structural rule $\tau$.

\begin{theorem}\label{thm:evostability} \textbf{(Stability of Structurally Evolving Networks)}
Let $(F,X)$ be a dynamical network given by \eqref{eq:netclass} and $\tau$ a structural rule. The evolved dynamical network $(F_{\tau},X_{\tau})$ is intrinsically stable if and only if $(F,X)$ is intrinsically stable.
\end{theorem}

The importance of theorem \ref{thm:evostability} is that it describes a general mechanism for evolving the structure of a network that preserves the network's stability. On one hand this has potential applications to network design as one could create a large network that is dynamically stable by designing a much smaller network and evolving it over any number of rules. On the other hand, one could investigate what structural rules can be used to model the growth of certain types of networks, specifically those networks that preserve a distinct function as they grow over time, e.g. biological networks including neural, gene regulatory, protein-protein interaction, and metabolic networks. The goal in this case would be to discover what structural rules model the growth of such networks.

A proof of theorem \ref{thm:evostability} is the following.

\begin{proof}
Suppose the dynamical network $(F,X)$ is intrinsically stable so that in particular $\rho(\Lambda)<1$. Given a structural rule $\tau$, theorem \ref{thm1} stated in terms of matrices implies that
\begin{equation}\label{eq:matrixver}
\sigma(\tau(\Lambda))=\sigma(\Lambda)\cup\sigma(\Lambda_1)^{n_1-1}\cup\dots\cup\sigma(\Lambda_m)^{n_m-1}
\end{equation}
where $\Lambda_i$ are square submatrices of $\Lambda$ and $n_i\geq 0$ for all $1\leq i\leq m$.

Since $\Lambda$ is a nonnegative matrix then $\rho(\Lambda_i)\leq\rho(\Lambda)$ for all $1\leq i\leq m$ (See, for instance, corollary 8.1.20 in \cite{HJ90}). Hence, $\rho(\tau(\Lambda))=\rho(\Lambda)<1$ by equation \eqref{eq:matrixver}. Since $\tau(\Lambda)=\Lambda_{\tau}$ is the stability matrix of $(F_{\tau},X_{\tau})$ then this implies that $\rho(\Lambda_{\tau})<1$ so that the evolved network $(F_{\tau},X_{\tau})$ is intrinsically stable. Conversely, if $(F_{\tau},X_{\tau})$ is intrinsically stable then $\rho(\Lambda_{\tau})<1$ and equation \eqref{eq:matrixver} immediately implies that $\rho(\Lambda)<1$, completing the proof.
\end{proof}

For the recurrent network $(R,\mathbb{R}^{4})$ in example \ref{ex:DRNN} the stability matrix $\Lambda$ has eigenvalues $\sigma(\Lambda)=\{|a|\pm\sqrt{2}|b|,|a|,|a|\}$. Hence, $(R,\mathbb{R}^{4})$ is intrinsically stable if $||a|\pm\sqrt{2}|b||<1$. If this condition holds then, theorem \ref{thm:evostability} implies that the evolved network $(R_{\delta},\mathbb{R}^6)$ is also intrinsically stable. Here, one can directly compute that $\sigma(\Lambda_{\delta})=\{|a|\pm\sqrt{2}|b|\}\cup\{|a|\}^4$ verifying the result.

It is worth noting that if $(F,X)$ is given by \eqref{eq:netclass} then its evolution $(F_{\tau},X_{\tau})$ is also of the same form. Hence, $(F_{\tau},X_{\tau})$ can also be evolved by $\tau$, which results in the dynamical network $(F_{\tau^2},X_{\tau^2})$. As a direct consequence to theorem \ref{thm:evostability}, if $(F,X)$ is intrinsically stable then any sequence of evolutions of $(F,X)$ results in an intrinsically stable network.

\begin{corollary}\label{cor:1}\textbf{(Sequential Network Evolutions)}
Suppose $(F,X)$ is a dynamical network given by \eqref{eq:netclass} and $\tau$ is a structural rule. If $(F,X)$ is intrinsically stable then $(F_{\tau^j},X_{\tau^j})$ is intrinsically stable for all $j\geq 0$.
\end{corollary}

Continuing the sequence of network evolutions in example \ref{ex:DRNN}, if we evolve the network $(R_{\delta},\mathbb{R}^{6})$ again with respect to $\delta$ the result is the dynamical network $(R_{\delta^2},\mathbb{R}^{10})$ whose stability matrix $\Lambda_{\delta^2}$ is represented by the graph $\delta^2(G)$ shown in figure \ref{fig0} (right). Here, one can check that $$\sigma(\Lambda_{\delta^2})=\{a\pm\sqrt{2}b\}\cup\{a\}^8.$$
As guaranteed by corollary \ref{cor:1}, $(R_{\tau^2},\mathbb{R}^{10})$ is intrinsically stable if and only if the original network $(R,\mathbb{R}^4)$ is also intrinsically stable.

Similar to graph evolutions, one can generalize the notion of sequentially evolving a dynamical network $(F,X)$ over a single rule to $\tau$ to sequentially evolving $(F,X)$ over the sequence $\tau_1$, $\tau_2$, $\tau_3,\dots$ where each $\tau_i$ is a structural rule. The result is the sequence of dynamical network's
\[
(F,X), \ (F_{\bar{\tau}_1},X_{\bar{\tau}_1}), \ (F_{\bar{\tau}_2},X_{\bar{\tau}_2}), \ (F_{\bar{\tau}_3},X_{\bar{\tau}_3}), \dots
\]
where $(F_{\bar{\tau}_i},X_{\bar{\tau}_i})$ is the dynamical network $(F,X)$ sequentially evolved over the rules $\tau_1,\tau_2\dots,\tau_i$.

It is worth mentioning that, since the rule $\tau$ used in theorem \ref{thm:evostability} is arbitrary, evolving $(F,X)$ over any sequence of rules $\tau_1$, $\tau_2$, $\tau_3,\dots$ will not destabilize the network if it is intrinsically stable. In fact it follows that, if $(F,X)$ is intrinsically stable then $(F_{\bar{\tau}_j},X_{\bar{\tau}_j})$ is intrinsically stable for all $j\geq 0$, which is a generalization of corollary \ref{cor:1}.

In contrast, if a dynamical network $(F,X)$ is stable but not intrinsically stable, it can fail to maintain its stability as its topology evolves even if it is a simple linear dynamical network as is illustrated in the following example.

\begin{example}\label{ex:loss} \textbf{(Loss of Stability)}
Let $\tau$ be the rule that selects all vertices of a graph without loops. Consider the linear dynamical network $(F,\mathbb{R}^3)$ given by
\[
F(\mathbf{x})=
\left[\begin{array}{rrr}
0&-1&3/4\\
0&0&1/2\\
-1/2&0&3/2
\end{array}\right]\left[\begin{array}{r}
x_1\\
x_2\\
x_3
\end{array}\right].
\]
Its evolution $(F_{\tau}, \mathbb{R}^4)$ with respect to $\tau$ is given by
\begin{equation}\label{eq:first}
F_{\tau}(\mathbf{x})=
\left[\begin{array}{rrrr}
0&-1&0&3/4\\
0&0&1/2&0\\
-1/2&0&-3/2&0\\
-1/2&0&0&-3/2
\end{array}\right]
\left[\begin{array}{r}
x_1\\
x_2\\
x_3\\
x_4
\end{array}\right].
\end{equation}
If $A\in\mathbb{R}^{3\times 3}$ is the matrix in \eqref{eq:first} such that $F(\mathbf{x})=A\mathbf{x}$ then $F_{\tau}(\mathbf{x})=\tau(A)\mathbf{x}$ where $\tau(A)\in\mathbb{R}^{4\times 4}$. Here one can compute that $\rho(A)\approx.938$ whereas $\rho(\tau(A))=3/2$. Since both systems are linear, it follows immediately that $(F,\mathbb{R}^3)$ is a stable dynamical network  whereas $(F_{\tau},\mathbb{R}^4)$ is unstable.

The reason $(F,\mathbb{R}^3)$ can lose its stability as it is evolved is that it is not intrinsically stable. That is, the stability matrix of $(F,X)$ is the matrix $|A|$, which is the matrix with entries $|A|_{ij}=|A_{ij}|$ with spectral radius $\rho(|A|)=1.787$. Since this is greater than one, the system is stable but not intrinsically stable. Therefore, it is possible, as it is demonstrated here, for the network to lose stability as its topology evolves.
\end{example}

Example \ref{ex:loss} is meant to emphasize the fact that more than the standard notion of stability is needed to guarantee a network's stability as a network evolves under some rule $\tau$. This together with the results of theorem \ref{thm:evostability} suggest that networks, possibly even real networks, need to maintain a stronger version of stability in order to preserve their function as their structure evolves under a structural rule.

\section{Conclusion}\label{conc}

In this paper we have described a method that evolves the topology of a network in a way that preserves both the spectral and local structure of the network. This method, which we refer to as a topological evolution of a network, is quite flexible in that the topology of any network can be evolved around any network core, which is any subset of the network elements.

As the results of an evolution depend on the particular subset or network core that is used, this method has the potential to model growth in a variety of networks. What is important in modeling the growth of a specific network is determining a particular subset of network elements that the network will evolve around. As this can be any subset of the network's elements there is a seemly unending number of ways in which a network can be evolved. In this sense the expert, e.g. biologist, sociologist, computer scientist, is needed to pick the set of elements that can be used to best model the growth of the particular biological, social, or technological network under consideration.

Beside modeling network growth this technique of evolving a network's topology can also be used to compare the topology of different networks. That is, two networks can have very different topologies but when evolved with respect to some structural rule $\tau$, the resulting networks may have the same topology. If this is the case, we say the two networks are evolution equivalent meaning the two networks are similar with respect to the rule $\tau$. In this way, the rule $\tau$ allows those studying a particular class of networks a way of comparing the \emph{evolved topology} of these networks and drawing conclusions about both the evolved and original networks. Of course, it is again important that this rule be designed by an expert to have some significance with respect to the nature of the network(s) being considered. The main idea we put forth here is that many networks currently under study are likely to have features that come to light only as these networks are evolved.

Because this method of evolving a network's topology also preserves the network's spectral structure it can also be used to study the interplay of network growth and function. The reason is that the dynamics of a network is related to the network's spectrum, the network's dynamics is in turn related to how well the network is able to perform specific tasks. In particular, a network's growth can have a destabilizing effect on the network's dynamics, which can lead to network failure.

We show that if a network's dynamics is \emph{intrinsically stable} then the network will remain intrinsically stable even as the network's topology evolves, i.e. the growth of the network will not change the network's dynamics, at least not qualitatively. We note that the notion of intrinsic stability has been previously studied in the context of dynamical networks with time delays, were is was shown that an intrinsically stable dynamical network will remain intrinsically stabile even if time delays are introduced into or removed from the network's interactions \cite{BW13}. Hence, networks that are intrinsically stable are dynamically resilient to both changes in the network's topology and changes in the network's environment that cause time delays. Moreover, because it is straightforward to verify whether a network is intrinsically stable, this notion of stability has potential to be both a practical and useful tool in the design of dynamically stable networks.

\section{Appendix: Isospectral Graph Reductions}\label{appendix}

To prove theorem \ref{thm1} we need the notion of an \emph{isospectral graph reduction}, which is a way of reducing the size of a graph while essentially preserving its set of eigenvalues. Since there is a one-to-one relation between the graphs we consider and the matrices $M\in\mathbb{R}^{n\times n}$, there is also an equivalent theory of \emph{isospectral matrix reductions}. Both types of reductions will be useful to us.

For the sake of simplicity we begin by defining an isospectral matrix reduction. For these reductions we need to consider matrices of rational functions. The reason is that, by the Fundamental Theorem of Algebra, a matrix $A\in\mathbb{R}^{n\times n}$ has exactly $n$ eigenvalues including multiplicities. In order to reduce the size of a matrix while at the same time preserving its eigenvalues we need something that carries more information than scalars. The objects we will use are rational functions. The specific reasons for using rational functions can be found in \cite{BWBook}, chapter 1.

We let $\mathbb{W}^{n\times n}$ be the set of $n\times n$ matrices whose entries are rational functions  $p(\lambda)/q(\lambda)\in\mathbb{W}$, where $p(\lambda)$ and $q(\lambda)\neq0$ are polynomials with real coefficients in the variable $\lambda$. The eigenvalues of the matrix $M=M(\lambda)\in\mathbb{W}^{n\times n}$ are defined to be solutions of the \emph{characteristic equation}
\[
\det(M(\lambda)-\lambda I)=0,
\]
which is an extension of the standard definition of the eigenvalues for a matrix with complex entries.

For $M\in\mathbb{R}^{n\times n}$ let $N=\{1,\ldots,n\}$. If the sets $R,C\subseteq N$ are proper subsets of $N$, we denote by $M_{RC}$ the $|R| \times |C|$ \emph{submatrix} of $M$ with rows indexed by $R$ and columns by $C$. The isospectral reduction of a square real valued matrix is defined as follows.

\begin{definition}\label{def:isored} \textbf{(Isospectral Matrix Reduction)}
The \emph{isospectral reduction} of a matrix $M\in\mathbb{R}^{n\times n}$ over the proper subset $S\subseteq N$ is the matrix
\[
\mathcal{R}_S(M) = M_{SS} - M_{S\bar{S}}(M_{\bar{S}\bar{S}}-\lambda I)^{-1} M_{\bar{S}S}\in\mathbb{W}^{|S|\times|S|}.
\]
\end{definition}

The relation between the eigenvalues of the matrix $M$ and its isospectral reduction $\mathcal{R}_S(M)$ is described by the following theorem \cite{BW12,BWBook}.

\begin{theorem}\label{thm:maintheorem}\textbf{(Spectrum of Isospectral Reductions)} For $M\in\mathbb{W}^{n\times n}$ and the proper set $S\subset N$ the eigenvalues of the isospectral reduction $\mathcal{R}_S(M)$ are
$$\sigma\big(\mathcal{R}_S(M)\big)=\sigma(M)-\sigma(M_{\bar{S}\bar{S}}).$$
\end{theorem}

Phrased in terms of graphs, if the graph $G=(V,E,\omega)$ is isospectrally reduced over some proper subset of its vertices $S\subset V$ then the resulting reduced graph $\mathcal{R}_S(G)=(\mathcal{V},\mathcal{E},\mu)$ with rational function weights has the eigenvalues
\[
\sigma\big(\mathcal{R}_S(G)\big)=\sigma(G)-\sigma(G|\bar{S}).
\]
It is worth noting that the graph $G$ and its subgraph $G|\bar{S}$ may have no eigenvalues in common, in which case $\sigma(\mathcal{R}_S(G))=\sigma(G)$. However, for the proof of theorem \ref{thm1} we will be using the fact that an evolved graph $\mathcal{X}_S(G)$ and its restriction $\mathcal{X}_S(G)|\bar{S}$ have eigenvalues in common.

A proof of theorem \ref{thm1} is the following.

\begin{proof}
For the graph $G=(V,E,\omega)$ let $S=\{v_1,\dots,v_\ell\}$ where $V=\{v_1,\dots,v_n\}$. By a slight abuse in notation we also let $S$ denote the index set $S=\{1,\dots,\ell\}$ that indexes the vertices in $S$, so that the reduction $\mathcal{R}_S(G)=\mathcal{R}_S(M)$ where $M=M(G)$.

For the moment we assume that the graph $G|\bar{S}$ has the single strongly connected component $C_1$. The weighted adjacency matrix $M\in\mathbb{R}^{n\times n}$ then has the block form
\[
M=
\left[\begin{array}{cc}
U&W\\
Y&Z
\end{array}\right]
\]
where $U\in\mathbb{R}^{\ell\times\ell}$ is the matrix $U=M_{SS}$, which is the weighted adjacency matrix of $G|S$. The matrix $Z\in\mathbb{R}^{n-\ell\times n-\ell}$ is the matrix $Z=M_{\bar{S}\bar{S}}$, which is the weighted adjacency matrix of $G|\bar{S}=C_1$. The matrix $W=M_{S\bar{S}}\in\mathbb{R}^{\ell\times n-\ell}$ is the matrix of edges weights of edges from $G|S$ to $C_1$ and $Y=M_{\bar{S}S}\in\mathbb{R}^{n-\ell\times \ell}$ is the matrix of edge weights of edges from $C_1$ to $G|S$.

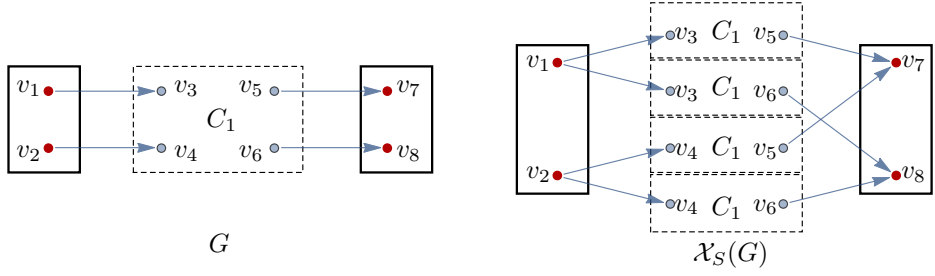
\begin{figure}
\begin{center}
\begin{tabular}{c}
    \begin{overpic}[scale=.43]{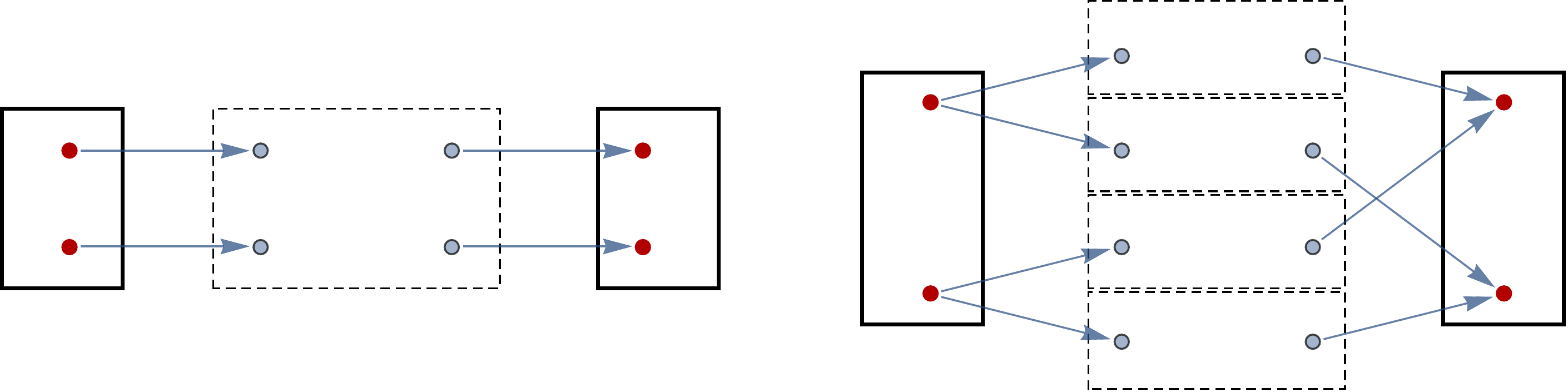}
    \put(21.75,-2){$G$}
    \put(21.5,11.5){$C_1$}
    \put(18,15.25){$v_3$}
    \put(18,8.25){$v_4$}
    \put(25,8.25){$v_6$}
    \put(25,15.25){$v_5$}
    \put(1,15.25){$v_1$}
    \put(1,8.25){$v_2$}
    \put(42,15.25){$v_7$}
    \put(42,8.25){$v_8$}

    \put(74,-3){$\mathcal{X}_S(G)$}
    \put(76,21){$C_1$}
    \put(76,15){$C_1$}
    \put(76,8.5){$C_1$}
    \put(76,2){$C_1$}
    \put(56,18){$v_1$}
    \put(56,6){$v_2$}
    \put(96.5,18){$v_7$}
    \put(96.5,6){$v_8$}

    \put(72,21){$v_3$}
    \put(72,14.75){$v_3$}

    \put(72,8.5){$v_4$}
    \put(72,2.5){$v_4$}

    \put(80.5,21){$v_5$}
    \put(80.5,14.75){$v_6$}

    \put(80.5,8.5){$v_5$}
    \put(80.5,2.5){$v_6$}

    \end{overpic}
\end{tabular}
\end{center}
  \caption{An example of a graph $G=(V,E,\omega)$ with the single strongly connected component $C_1=G|\bar{S}$ is shown (left), where solid boxes indicate the graph $G|S$. As there are two edges from $G|S$ to $C_1$ and two edges from $C_1$ to $G|S$ there are $2\times 2$ branches in $\mathcal{B}_S(G)$ containing $C_1$. These are merged together with $G|S$ to form $\mathcal{X}_S(G)$ (right).}\label{fig00}
\end{figure}

The evolution $\mathcal{X}_S(G)$ is the graph in which all component branches of the form
\[
\beta=v_i,e_{ip},C_1,e_{qj},v_j \ \ \text{for all} \ \ v_i,v_j\in S, v_p,v_q\in \bar{S} \ \ \text{and} \ \ e_{ip},e_{qj}\in E
\]
are merged together with the graph $G|S$ (see figure \ref{fig00}). The weighted adjacency matrix $\hat{M}=M(\mathcal{X}_S(G))$ has the block form
\[
\hat{M}=
\left[\begin{array}{cccc}
U&\Big[W_1 \hspace{0.15in} \cdots\hspace{0.15in} W_1\Big]&\cdots&\Big[W_{\beta} \hspace{0.15in} \cdots \hspace{0.15in} W_{\beta}\Big]\\\\

\left[\begin{array}{c}
Y_1\\
\vdots\\
Y_\gamma
\end{array}\right]&

\left[\begin{array}{ccc}
Z&&\\
&\ddots&\\
&&Z
\end{array}\right]&&0\\
\vdots&&\ddots&\\
\left[\begin{array}{c}
Y_1\\
\vdots\\
Y_\gamma
\end{array}\right]&0&&

\left[\begin{array}{ccc}
Z&&\\
&\ddots&\\
&&Z
\end{array}\right]\\

\end{array}\right]=
\left[\begin{array}{cc}
U&\hat{W}\\
\hat{Y}&\hat{Z}
\end{array}\right],
\]
where each $W_i\in\mathbb{R}^{\ell\times n-\ell}$, each $Y_j\in\mathbb{R}^{n-\ell\times \ell}$, $\sum_{i=1}^w W_i=W$ and $\sum_{j=1}^y Y_j=Y$. Here, $w\geq0$ is the number of directed edges from $G|S$ to $C_1$. The matrix $W_i$ has a single nonzero entry corresponding to exactly one edge from this set of edges. Similarly, $y\geq0$ is the number of directed edges from $C_1$ to $G|S$. The matrix $Y_i$ has a single nonzero entry corresponding to exactly one edge from this set of edges. Since there are $w\cdot y$ component branches in $\mathcal{X}_S(G)$ containing $C_1$ then the matrix $\hat{M}\in\mathbb{R}^{n+(n_1-1)\ell\times n+(n_1-1)\ell}$ where $n_1=w\cdot y$.

The claim is that by reducing both $M$ and $\hat{M}$ over $S$ the result is the same matrix. To see this note that by theorem \ref{thm:maintheorem} the reduced matrix $\mathcal{R}_S(M)$ is
\[
\mathcal{R}_S(M)=U-W(Z-\lambda I)^{-1}Y\in\mathbb{W}^{|S|\times|S|}.
\]
For the matrix $\hat{M}$ its reduction over $S$ is the matrix
\begin{align*}
\mathcal{R}_S(\hat{M})&=U-\hat{W}(\hat{Z}-\lambda I)^{-1}\hat{Y}\\
&=U-\hat{W} \ \text{diag}[(Z-\lambda I)^{-1},\dots,(Z-\lambda I)^{-1}]\hat{Y}\\
&=U-\sum_{i=1}^w\sum_{j=1}^y W_i(Z-\lambda I)^{-1}Y_j\\
&=U-\big(\sum_{i=1}^w W_i\big)(Z-\lambda I)^{-1}\big(\sum_{j=1}^y Y_j\big)\\
&=U-W(Z-\lambda I)^{-1}Y\in\mathbb{W}^{|S|\times|S|}.
\end{align*}
This verifies the claim that $\mathcal{R}_S(M)=\mathcal{R}_S(\hat{M})$.

Theorem \ref{thm:maintheorem} then implies that $\sigma(M)-\sigma(M_{\bar{S}\bar{S}})=\sigma(\hat{M})-\sigma(\hat{M}_{\bar{S}\bar{S}})$. Since $\sigma(M_{\bar{S}\bar{S}})=\sigma(C_1)$ and $\sigma(\hat{M}_{\bar{S}\bar{S}})=\sigma(C_1)^{n_1}$ we then have
\[
\sigma(G)-\sigma(C_1)=\sigma(\mathcal{X}_S(G))-\sigma(C_1)^{n_1}.
\]
Since $G$ has $n$ eigenvalues, $\mathcal{X}_S(G)$ has $n+(n_1-1)\ell$, and $C_1$ has $\ell$ eigenvalues respectively including multiplicities, it follows that
\begin{equation}\label{eq:last}
\sigma(\mathcal{X}_S(G))=\sigma(G)\cup\sigma(C_1)^{n_1-1}
\end{equation}
so that theorem \ref{thm1} holds in the case that $G|\bar{S}$ has a single strongly connected component $C_1$.

If $C_1,\dots,C_m$ are the components of the restricted graph $G|\bar{S}$ where $m>1$ then we continue inductively. Let $S_i\subset V$ be the vertices that do not belong to $C_i$ or $i=1,\dots, m$. Then the evolution $\mathcal{X}_{S_1}(G)$ has eigenvalues
\[
\sigma(\mathcal{X}_{S_1}(G))=\sigma(G)\cup\sigma(C_1)^{n_1-1}
\]
by equation \eqref{eq:last}, where $n_1$ is the number of component branches in $\mathcal{B}_{S_1}(G)$ containing $C_1$. Since
\[
\mathcal{X}_{S_1-\bar{S}_2}(G)=\mathcal{X}_{S_1}\Big(\mathcal{X}_{S_2}\big(\mathcal{X}_{S_1}(G)\big)\Big)
\]
then by repeated use of the same argument that
\[
\sigma(\mathcal{X}_{S_1-\bar{S}_2}(G))=\sigma(G)\cup\sigma(C_1)^{n_1-1}\cup\sigma(C_2)^{n_2-1},
\]
where $n_1$ and $n_2$ are the number of component branches in $\mathcal{B}_{S_1-\bar{S}_2}(G)$ containing $C_1$ and $C_2$, respectively. Continuing in this manner it follows that
\[
\sigma(\mathcal{X}_{S}(G))=\sigma(G)\cup\sigma(C_1)^{n_1-1}\cup\dots\cup\sigma(C_m)^{n_m-1},
\]
where $n_i$ is the number of component branches in $\mathcal{B}_{S}(G)$ containing $C_i$ for all $i=1,\dots,m$; since $S_1-\cup_{i=2}^m\bar{S}_i=S$. This completes the proof.
\end{proof}

A proof of proposition \ref{prop:0} is based on the following result relating the eigenvectors of the a graph $G$ and its reduction $\mathcal{R}_S(G)$.

\begin{theorem}\label{thm:reduction}\textbf{(Eigenvectors of Reduced Matrices)}
Suppose $M\in\mathbb{R}^{n\times n}$ and $S\subseteq N$. If $(\lambda,\mathbf{v})$ is an eigenpair of $M$ and $\lambda\notin \sigma(M_{\bar{S}\bar{S}})$ then $(\lambda,\mathbf{v}_S)$ is an eigenpair of $\mathcal{R}_S(M)$.
\end{theorem}

\begin{proof}
Suppose $(\lambda,\mathbf{v})$ is an eigenpair of $M$ and $\lambda\notin \sigma(M_{\bar{S}\bar{S}})$. Then without loss in generality
we may assume that $\mathbf{v}=(\mathbf{v}_ S^T,\mathbf{v}_{\bar{S}}^T)^T$. Since $M\mathbf{v}=\lambda\mathbf{v}$ then
\[
\left[\begin{array}{cc}
M_{SS}&M_{S\bar{S}}\\
M_{\bar{S}S}&M_{\bar{S}\bar{S}}
\end{array}\right]
\left[\begin{array}{c}
\mathbf{v}_S\\
\mathbf{v}_{\bar{S}}
\end{array}\right]=
\lambda \left[\begin{array}{c}
\mathbf{v}_S\\
\mathbf{v}_{\bar{S}}
\end{array}\right],
\]
which yields two equations the second of which implies that
\[
M_{\bar{S}S}\mathbf{v}_S+M_{\bar{S}\bar{S}}\mathbf{v}_S=\lambda\mathbf{v}_{\bar{S}}.
\]
Solving for $\mathbf{v}_{\bar{S}}$ in this equation yields
\begin{equation}\label{eq:bar}
\mathbf{v}_{\bar{S}}=-(M_{\bar{S}\bar{S}}-\lambda I)^{-1}M_{\bar{S}S}\mathbf{v}_{ S},
\end{equation}
where $M_{\bar{S}\bar{S}}-\lambda I$ is invertible given that $\lambda\notin \sigma(M_{\bar{S}\bar{S}})$.

Note that
\begin{align*}
(M-\lambda I)\mathbf{v}&=
\left[\begin{array}{c}
(M-\lambda I)_{ SS}\mathbf{v}_{ S}+(M-\lambda I)_{ S\bar{S}}\mathbf{v}_{\bar{S}}\\
(M-\lambda I)_{\bar{S}S}\mathbf{v}_{ S}+(M-\lambda I)_{\bar{S}\bar{S}}\mathbf{v}_{\bar{S}}
\end{array}\right]\\
&=
\left[\begin{array}{c}
M_{ SS}\mathbf{v}_{ S}-M_{ S\bar{S}}(M_{\bar{S}\bar{S}}-\lambda I)^{-1}M_{\bar{S}S}\mathbf{v}_{ S}\\
M_{\bar{S}S}\mathbf{v}_{ S}-(M_{\bar{S}\bar{S}}-\lambda I)(M_{\bar{S}\bar{S}}-\lambda I)^{-1}M_{\bar{S}S}\mathbf{v}_{ S}
\end{array}\right]\\
&=\left[\begin{array}{c}
(\mathcal{R}_S(M)-\lambda I)\mathbf{v}_{ S}\\
0
\end{array}\right].
\end{align*}
Since $(M-\lambda I)\mathbf{v}=0$ it follows that $(\lambda,\mathbf{v}_S)$ is an eigenpair of $\mathcal{R}_S(M)$.

Moreover, we observe that if $(\lambda,\mathbf{v}_S)$ is an eigenpair of $\mathcal{R}_S(M)$ then by reversing this argument, $\big(\lambda,(\mathbf{v}_ S^T,\mathbf{v}_{\bar{S}}^T)^T\big)$ is an eigenpair of $M$ where $\mathbf{v}_{\bar{S}}$ is given by \eqref{eq:bar}.
\end{proof}

We now give a proof of proposition \ref{prop:0}.

\begin{proof}
Let $M=M(G)$ and $\hat{M}=M(\mathcal{X}_S(G))$ where $S\subseteq N$. If $(\lambda,\mathbf{v})$ is an eigenpair of $M$ and $\lambda\notin \sigma(M_{\bar{S}\bar{S}})$ then theorem \ref{thm:reduction} implies that $(\lambda,\mathbf{v}_S)$ is an eigenpair of $\mathcal{R}_S(M)$. Using the fact that $\mathcal{R}_S(\mathcal{X}_S(M))=\mathcal{R}_S(M)$ (see the proof of theorem \ref{thm1}) and the observation in the last line of the proof of theorem \ref{thm:reduction} it follows that $(\lambda,\hat{\mathbf{v}})$ is an eigenpair of $\hat{M}$ where
\[
\hat{\mathbf{v}}=
\left[
\begin{array}{c}
\hat{\mathbf{v}}_S\\
\hat{\mathbf{v}}_{\bar{S}}
\end{array}
\right]=
\left[
\begin{array}{c}
\mathbf{v}_S\\
-(\hat{M}_{\bar{S}\bar{S}}-\lambda I)^{-1}\hat{M}_{\bar{S}S}\mathbf{v}_{ S}
\end{array}
\right].
\]
Note that $\mathbf{v}_S=\hat{\mathbf{v}}_S$, which completes the proof.
\end{proof}

We now give a proof of proposition \ref{prop10}.

\begin{proof}
Suppose that $G=(V,E)$ is strongly connected and $S\subset V$. Since the $\mathcal{X}_S(G)$ preserves the path structure of $G$, i.e. there is a path from $v_i$ to $v_j$ in $\mathcal{X}_S(G)$ if and only if there is a path from the corresponding $v_i$ to $v_j$ in $G$, then $\mathcal{X}_S(G)$ must be strongly connected. Therefore, both $G$ and $\mathcal{X}_S(G)$ have eigencentrality vectors.

Given that $M(G)$ is a nonnegative matrix, theorem \ref{thm1} together with corollary 8.1.20 in \cite{HJ90} imply that $G$ and $\mathcal{X}_S(G)$ have the same spectral radius $\rho$. Since $\rho$ is a simple eigenvalue of both $G$ and $\mathcal{X}_S(G)$, proposition \ref{prop:0} implies that given an eigencentrality vector $\mathbf{p}$ of $G$ there is an eigencentrality vector $\mathbf{q}$ of $\mathcal{X}_S(G)$ such that $\mathbf{p}_S=\mathbf{q}_S$ completing the proof.
\end{proof}

\section{Acknowledgement} The work of L. A. Bunimovich was partially supported by the NSF grant DMS-1600568

\end{document}